\documentclass[12pt,reqno]{amsart}
\usepackage[english]{babel}
\usepackage{soul}
\usepackage[T1]{fontenc}
\usepackage{amsmath,amssymb,amsfonts,amsthm}
\usepackage{hyperref,comment}
\usepackage{graphicx,color,stmaryrd}
\usepackage[dvipsnames]{xcolor}
\usepackage[margin=3cm]{geometry}
\usepackage[bbgreekl]{mathbbol}
\newcommand{\llb}{\llbracket}
\newcommand{\rrb}{\rrbracket}
\usepackage{csquotes}
    \usepackage{dutchcal}

\usepackage{mathtools}
\definecolor{light-gray}{gray}{0.95}
\usepackage[color=light-gray]{todonotes}
\usepackage[all,cmtip]{xy} 
\usepackage{pdfpages}
\DeclareMathAlphabet{\mathpzc}{OT1}{pzc}{m}{it}
\DeclareMathAlphabet\mathbfcal{OMS}{cmsy}{b}{n}

\usepackage[backend=biber, giveninits=true, style=alphabetic, sorting=nyt,isbn=false, maxalphanames=5,url=false, maxbibnames=99]{biblatex}
\newtheorem{definition}{{Definition}}[section]
\newtheorem{theorem}[definition]{{Theorem}}
\newtheorem{proposition}[definition]{{Proposition}}

\newtheorem{corollary}[definition]{Corollary}

\theoremstyle{definition}
\newtheorem{remark}[definition]{{Remark}}

\newtheorem{example}[definition]{{Example}}

\newcommand{\qsp}[2]{\,\ensuremath{\raise.5ex\hbox{$#1$}\big\slash\raise-.5ex\hbox{$#2$}}} 

\newcommand{\pard}[2]{\frac{\delta#1}{\delta#2}}

\newcommand{\intl}{\int\limits}

\newcommand{\tQ}{\widetilde{Q}}

\newcommand{\oloc}{\Omega_{\text{loc}}}
\newcommand{\iloc}{\Omega_{\int}}
\newcommand{\osrc}{\Omega_{\text{src}}}

\newcommand{\bbR}{\mathbb{R}}
\newcommand{\Lie}{\mathpzc{L}}

\newcommand{\Vol}{\mathrm{vol}}
\newcommand{\Ric}{\mathrm{Ric}}
\newcommand{\Lor}{\mathrm{Lor}}
\DeclareMathOperator{\Div}{div}

\newcommand{\Grad}{\mathrm{grad}}

\newcommand{\Hide}[1]{}
\newcommand{\sfBFV}{\footnotesize\mathsf{BFV}}
\newcommand{\sfBV}{\footnotesize\mathsf{BV}}

\newcommand{\calF}{\mathcal{F}}
\newcommand{\Fcl}{\calF}
\newcommand{\calFBV}{\calF_{\sfBV}}
\newcommand{\Fclp}{\calF^\partial}
\newcommand{\cFclp}{\check{\calF}^\partial}
\newcommand{\Scl}{S}
\newcommand{\SBV}{S_{\sfBV}}
\newcommand{\OmegaBV}{\Omega_{\sfBV}}

\newcommand{\calX}{\mathcal{X}}

\title{A Lie-Rinehart algebra in general relativity}

\author{Christian Blohmann}
\address{Max Planck Institute for Mathematics, Vivatsgasse 7, 53111 Bonn, Germany}
\email{blohmann@mpim-bonn.mpg.de}

\author{Michele Schiavina}
\address{Dipartimento di Matematica, Università di Pavia, via Ferrata 5, 27100 Pavia, Italy}

\email{michele.schiavina@unipv.it}

\author{Alan Weinstein}
\address{Department of Mathematics, University of California, Berkeley, CA 94720 USA and  Department of Mathematics, Stanford University, Stanford, CA 94305
USA}
\email{alanw@math.berkeley.edu}
\begin{document}

\thanks{M.S. acknowledges support from the NCCR SwissMAP, funded by the Swiss National Science Foundation. Part of the research has been carried out while M.S. was a visiting scientist at the Max Planck Institute for Mathematics, Bonn. }

\begin{abstract}
We construct a Lie-Rinehart algebra over an infinitesimal extension of the space of initial value fields for Einstein's equations.  The bracket relations in this algebra are precisely those of the constraints for the initial value problem.  The Lie-Rinehart algebra comes from a slight generalization of a Lie algebroid in which the algebra consists of sections of a sheaf rather than a vector bundle.  (An actual Lie algebroid had been previously constructed by Blohmann, Fernandes, and Weinstein over a much larger extension.)  The construction uses the BV-BFV (Batalin-Fradkin-Vilkovisky) approach to boundary value problems, starting with the Einstein equations themselves, to construct an $L_\infty$-algebroid over a graded manifold which extends the initial data.  The Lie-Rinehart algebra is then constructed by a change of variables.  One of the consequences of the BV-BFV approach is a proof that the coisotropic property of the constraint set follows from the invariance of the Einstein equations under space-time diffeomorphisms.
\end{abstract}

\dedicatory{Dedicated to Victor Guillemin, whose work has inspired us for many years.}
\maketitle

\tableofcontents

\section{Introduction} 
In \cite{BFW}, two of the present authors and Marco Fernandes found a Lie algebroid whose bracket relations on constant sections agree precisely with the Poisson bracket relations among the energy and momentum constraints for the initial value problem of Einstein's equations in general relativity.  Whereas the phase space for the initial value problem is the cotangent bundle of the space $\mathbfcal{R}$ of riemannian metrics on a Cauchy hypersurface $\Sigma$, the base of the ``Lie algebroid of evolutions'' in \cite{BFW} is the much larger space of paths $[0,1]\to \mathbfcal{R}.$  It was observed there that a natural idea for reducing the base to the cotangent bundle\footnote{We define the fibers of $T^*\mathbfcal{R}$ as the space of smooth, density-valued, contravariant symmetric 2-tensors on $\Sigma$.} $T^*\mathbfcal{R}$ produced an object with bracket and anchor which failed to satisfy the axioms of a Lie algebroid.

The original goal of \cite{BFW} was to explain the well known coisotropic property of the constraints' zero set by an extension to Lie algebroids of the theory of hamiltonian actions of Lie algebras on (pre)symplectic manifolds, for which the zero set of the momentum map is coisotropic.  Indeed, in \cite{BlohmannWeinstein:HamLA}, a suitable notion of hamiltonian Lie algebroid was developed for which the zero set of what was now a momentum {\em section} of the dual of the Lie algebroid was shown to be coisotropic.  It turns out, though, that the Lie algebroid of evolutions is not hamiltonian, at least for the natural presymplectic structure on the paths in $\mathbfcal{R}$ obtained by pulling back the canonical symplectic structure on $T^*\mathbfcal{R}$ by the map assigning to each path its initial value and normal derivative at $t=0.$\footnote{$T^*\mathbfcal{R}$ is identified with $T\mathbfcal{R}$ via a natural riemannian metric on $\mathbfcal{R}$.}

In this paper, we arrive in Section \ref{s:newLR} at an alternative realization of the constraint Poisson brackets by a generalized version of a Lie algebroid called a Lie-Rinehart algebra.  The base $\mathbfcal{B}$ of this object is no longer a  manifold as was that of the Lie algebroid of evolutions, but rather the product of $T^*\mathbfcal{R}$ with a one-point ringed space which can be seen as the first infinitesimal neighborhood of the origin in the product of vector fields and functions on $\Sigma$. The Lie-Rinehart algebra itself is a module over an algebra $\mathfrak{B}$ which plays the role of the functions on $\mathbfcal{B}$; this module carries a Lie algebra structure (over $\mathbb R$) for which the Leibniz rule for multiplication by elements of $\mathfrak{B}$ is specified by an ``anchor'' which is a map from the module into derivations of $\mathfrak{B}$.  The object of which this module are sections is not a vector bundle over $\mathbfcal{B}$, as it would be for a Lie algebroid, but rather a sheaf which is not locally constant.  The entire Lie-Rinehart structure is actually presented in terms of a cohomological vector field on a graded manifold lying over $\mathbfcal{B}$.  
Unfortunately, although the base of our Lie-Rinehart algebra is much smaller than the space of paths in $\mathbfcal{R}$, this structure still does not have the hamiltonian property (as generalized from Lie algebroids to Lie-Rinehart algebras); see Remark \ref{rmk:nonHamiltonianLRalg}. 
Nevertheless, we think that this construction is interesting because of its close connection with the invariance under diffeomorphisms of the Einstein equations.  

The machinery for our construction is the BV-BFV theory of \cite{CattaneoMnevReshetikhin2014,CMR18}, a combination of the Batalin--Vilkovisky and Batalin--Fradkin--Vilkovisky formalisms on manifolds with boundary \cite{BV1,BV2,BV3}, which associates supermanifolds carrying cohomological vector fields to boundary value problems for field theories with symmetry.  In Section \ref{s:gr}, we  review, and provide a new analysis of the output of, the application of the BV-BFV machinery to general relativity, as was presented by one of the authors in \cite{ScTH,CattaneoSchiavinaEH}. The result is an $L_\infty$-algebroid rather than Lie algebroid over a graded manifold which, despite being structurally similar to those arising in standard gauge-theoretic scenarios, is not associated to some (action) Lie algebroid.\footnote{See Remark \ref{rem:gaugeGRcomparison} to motivate this point of view.} 

We review the general BV-BFV theory itself in Section \ref{s:BVBFVconstructions}.  In particular, in Remark \ref{rmk:BFVcoisotropicexplanation1} and Corollary \ref{thm:BFVcoisotropicexplanation2} we argue that one can take the very existence of a BV-BFV structure as an explanation of the coisotropic property of the constraint set of general relativity without use of the formulas for the Poisson brackets of the constraint functions.

In Section \ref{s:newLR}, we eliminate the higher order terms in the bracket and anchor of the previously described $L_\infty$-algebroid by a change of variables which replaces  the odd variables in the (graded) base manifold  by infinitesimal variables.  At the same time, the vector bundle carrying our $L_\infty$ structure becomes a sheaf which is not locally trivial, thus necessitating the language of Lie-Rinehart algebras.

Section \ref{s:projectiontoinitialdata} is independent of Section \ref{s:newLR}.  In it, we try to construct a Lie algebroid over $T^*\mathbfcal{R}$ itself.  This time, we begin with the BFV data of Section \ref{s:gr} and restrict and project the homological vector field $Q_{\sfBFV}$ found there to the zero locus of the odd variables in the base of the $L_\infty$-algebroid constructed there.  The result is a vector field $Q_0$ whose square is no longer zero; in fact, it is zero only on the zero set $C$ of the constraint functions.  With the hindsight of knowing from the Poisson brackets of these functions that $C$ is coisotropic, we see that the object restricted to $C$, which {\em is} a Lie algebroid, is precisely the usual Lie algebroid attached to a coisotropic submanifold, essentially its conormal bundle, or, equivalently, its characteristic distribution.

We end the section with a discussion of several unsuccessful attempts (including by two of the present authors in \cite{BFW}) to restrict  the Lie algebroid of evolutions to $T^*\mathbfcal{R}$.

Finally, in Section \ref{s:Outlook}, we discuss some further possible ways to establish directly from diffeomorphism symmetry the fact that the constraint set is coisotropic.

\section*{Acknowledgements}
We would like to thank  S.\ D'Alesio, A.S.\ Cattaneo, N.L.\ Delgado, J. H\"ubschmann, L.\ Vitagliano, P.\ Xu, and M.\ Zambon for helpful discussions and for comments on earlier versions of the manuscript. M.S.\  would especially like to thank A.\ Riello for sharing comments that helped shape Section~\ref{s:unsuccessful}. Finally, we thank the referees for their helpful questions and suggestions.

\section{BV and BFV constructions}
\label{s:BVBFVconstructions}
In this section, we will outline two general frameworks for classical field theory that go under the acronyms BV and BFV, after Batalin, Fradkin and Vilkovisky in different collaborations \cite{batalin1983generalized,batalin1983quantization,batalin1984gauge}. Broadly speaking, one can think of the BFV formalism as the hamiltonian counterpart of the lagrangian BV formalism. They can be joined together when discussing field theory on manifolds with boundary to obtain what is often called the BV-BFV framework, following Cattaneo, Mnev and Reshetikhin \cite{CattaneoMnevReshetikhin2014}.

\subsection{Preliminaries}

The space of fields is given by the sections
\begin{equation*}
    \calF = \Gamma(M,F)
\end{equation*}
of a smooth fibre bundle $\pi_F: F \to M$. The space $\calF$ has a convenient smooth structure given by the functional diffeology for which smooth families, i.e.~plots in the language of diffeology, are smooth homotopies of sections.\footnote{We refer the reader to \cite{Iglesias-Zemmour:Diffeology} as a general reference on diffeology.} ($\calF$ can also be equipped with a Fr\'echet structure, but we will not need this here.) The diffeological structure suffices to define a notion of tangent bundle \cite{ChristensenWu:2016}. For $\calF$ it is given by 
\begin{equation*}
    T\calF = \Gamma(M,VF)
\end{equation*}
where $VF: \ker T\pi_F \to M$ is the vertical tangent bundle, viewed as a bundle over $M$.\footnote{In \cite{ChristensenWu:2016} this was shown to hold for $M$ compact, but the statement is also true for non-compact $M$.} If $F \to M$ is a vector bundle, then $\calF$ is a vector space and we have the natural trivialization $T\mathcal{F} \cong \mathcal{F}\times \mathcal{F}$.

The bundle projection $T\calF \to \calF$ maps a tangent vector given by a section $v: M \to VF$ to the field $\phi: M \to VF \to F$. The tangent space at the point $\phi \in T\calF$ is then given by (see \cite{Milnor63})
\begin{equation*}
  T_\phi \calF = \Gamma(M, \phi^* VF)
  \,,
\end{equation*}
where $\phi^* VF = M \times_F^{\phi,\pi_F} VF$ is the pullback of $VF \to F$ along $\phi: M \to F$.

A vector field on $\calF$ is a section $v: \calF \to T\calF$ of the bundle projection. It is called local, if for every $\phi \in \calF$ the value of $v(\phi) \in \Gamma(M,VF)$ at $m \in M$ depends only on the finite jet $j^k_m \phi$ for a fixed jet order $k$. A local vector field can be identified with the infinite prolongation of an evolutionary ``vector field'' (which is not a vector field) in the terminology of the calculus of variations (see \cite{Dedecker1950,Anderson:1989,AndersonTorre,Saunders:1989}). The space of local vector fields will be denoted by $\calX_{\text{loc}}(\calF)$.

There are two natural notions of differential forms on the diffeological space $\calF$. The first is by the left Kan extension of the de Rham functor from manifolds to diffeological spaces \cite{Iglesias-Zemmour:Diffeology}. This comes with a differential but without a natural notion of inner derivative. The second defines $k$-forms as fibre-wise alternating multilinear smooth functions $T^{(k)} \calF \to \bbR$ on the fibre product $T^{(k)}\calF\doteq T\calF \times_\calF \ldots \times_\calF T\calF$ of the tangent bundle. For such forms we have an inner derivative but no differential. While these two notions are different for general diffeological spaces, they do coincide for the space of fields $\calF$. 

Let $\Omega^{\bullet,\bullet}(\calF \times M)$ denote the bicomplex of differential forms. The differential in the direction of $\calF$ will be denoted by $\delta$, that in the direction of $M$ by $d$. A $(p,q)$-form can be viewed as a smooth fibre-wise multilinear map $T_p \calF \to \Omega^q(M)$, i.e.~a $p$-form on $\calF$ with values in $q$-forms on $M$. Both the  domain and codomain of this map are spaces of smooth sections of fibre-bundles over $M$. It then makes sense to call the form local if this map is local, i.e.~a differential operator. Local forms are a sub-bicomplex
\begin{equation*}
  \Omega^{\bullet,\bullet}_\mathrm{loc}(\calF \times M)
  \subset
  \Omega^{\bullet,\bullet}(\calF \times M)
  \,.
\end{equation*}

Let $j^\infty: \calF \times M \to J^\infty F$ be the infinite jet evaluation and let $\Omega^{\bullet,\bullet}(J^\infty F)$ be the variational bicomplex with vertical differential $\delta$ and horizontal differential $d$. It can be shown that the operation of pullback along $j^\infty$ is a morphism of bicomplexes with image
\begin{equation*}
  \Omega^{\bullet,\bullet}_\mathrm{loc}(\calF \times M)
  =
  (j^\infty)^* \Omega^{\bullet,\bullet}(J^\infty F)
  \,.
\end{equation*}
In this sense, the variational bicomplex can be viewed as the bicomplex of local forms on $\calF \times M$.\footnote{When $j^\infty$ is not surjective the pullback is not injective, so that we cannot simply identify the bicomplex with the complex of local forms. But this will not matter here.}

A form in $\Omega_\mathrm{loc}^{p,\text{top}}(\calF\times M)$ can be viewed as map $T_p \calF \to \Omega^\mathrm{top}(M)$. If $M$ is oriented and compact, we can compose this map with the integration over $M$, which yields a map $T_p \calF \to \bbR$. If $M$ is not compact, we can replace the integration with taking the $d$-cohomology class instead, which we denote suggestively by
\begin{equation*}
  \Omega_{\int}^p(\calF) 
  :=
  \Omega_\mathrm{loc}^{p,\mathrm{top}}(\calF \times M) /
  d \Omega_\mathrm{loc}^{p,\mathrm{top}-1}(\calF \times M)
  \,.
\end{equation*}
For $p=0$ and $M$ non-compact, an element of $\Omega_{\int}^0(\calF)$ can generally not be viewed as a function on $\calF$. For $p > 0$ the elements of $\Omega_{\int}^p(\calF)$ can be identified with special classes of forms in $\Omega^{p,\mathrm{top}}(\calF)$, called source forms for $p=1$ and functional forms for $p > 1$. A source is given in local coordinates by
\begin{equation*}
  \omega = \omega_\beta(x^1, u^\alpha, u^\alpha_{i_1}, \ldots, 
  u^\alpha_{i_1,\ldots, i_k}) 
  \delta u^\beta \wedge dx^1 \wedge \ldots \wedge dx^\mathrm{top}
  \,,
\end{equation*}
where $\omega_\beta$ is a function of the local bundle coordinates $x^1, u^\alpha$ and a finite number of the associated jet coordinates $u^\alpha_{i_1}, u^\alpha_{i_1,i_2}, \ldots$. The salient feature of a source forms is that it contains the vertical differential $\delta u^\alpha$ of the fibre coordinate, but not of higher jet coordinates. We denote the space of local forms of source type by $\Omega_\mathrm{src}^{1,\mathrm{top}}(\calF\times M)$.

In order to avoid the technical problems that arise from the non-compactness of $M$, we will from now on assume that $M$ is compact. Then the cotangent bundle $T^* \calF$ can be defined as follows. Let $V^*F$ denote the dual of the vertical vector bundle, viewed as a bundle over $M$. Let
\begin{equation*}
  T^* \calF := \Gamma\bigl(M, V^* F \otimes \mathrm{Dens}(M) \bigr)
  \,,
\end{equation*}
where $\mathrm{Dens}(M)$ denotes the density bundle on $M$.
This is a bundle over $\calF$ in the same way as the tangent bundle, the fibre over $\phi \in \calF$ being given by 
\begin{equation*}
  T^*_\phi \calF 
  = \Gamma\bigl(M, \phi^* V^* F \otimes \mathrm{Dens}(M) \bigr)
  \,.
\end{equation*}
The dual pairing between $v \in T_\phi \calF$ and $\alpha \in T^*_\phi \calF$ is given by the dual pairing of $VF$ and $V^* F$, which yields a density on $M$, followed by integration over $M$. For general considerations of dual vector spaces in field theory see for instance Sec.~3.5.5 and Appendix~B of \cite{CostelloGwilliam:FactorizationAlgebras1}.

\subsection{Lagrangian field theories and their boundary data}\label{s:LFTBoundary}
A lagrangian field theory (LFT) is specified by a space of fields $\Fcl$ together with a lagrangian (i.e.~a local $0,\text{top}$-form) $L\in\oloc^{0,\text{top}}(\calF\times M)$. We will be interested in the case of manifolds with boundary. For simplicity we will assume that $M$ is a cylinder of the form $M=\Sigma\times [0,1]$, where $\Sigma$ is a closed manifold. We also introduce the action function by integration $S=\intl_M L\in\iloc^{0}(\calF)$.\footnote{Recall that we assume $M$ to be compact.}

The physical fields of an LFT are those in the \emph{zero locus} $\mathsf{Z}(\mathsf{el})$ of the Euler--Lagrange one form, i.e.\ the set of solutions of the Euler--Lagrange equations of the variational problem associated to $S$. Since $S$ is represented by a local lagrangian $L\in \Omega_{\text{loc}}^{0,\text{top}}(\calF\times M)$, its variation can be decomposed as
\[
    \delta L = \mathsf{el} + d\gamma
\]
where $\mathsf{el}\in \osrc^{1,\text{top}}(\calF\times M)$ is a source form and $\gamma\in \Omega^{1,\text{top}-1}(\calF\times M)$. We denote by $\mathcal{I}_{EL}\subset C^\infty(\Fcl)$ the vanishing ideal associated to $\mathsf{Z}(\mathsf{el})$.

We can induce boundary data for a lagrangian field theory on a manifold with boundary $(M,\partial M)$. In order to do this, one first considers the space of germs of fields at the \emph{incoming} boundary $\Sigma_0 = \Sigma \times \{0\}$, denoted by $\cFclp_{\Sigma}$, which comes equipped with the surjective submersion $\check{\pi}_{0}\colon \Fcl\to{\cFclp_{\Sigma}}$ given by restriction, and a presymplectic structure $\check\omega$ induced by the variational problem
$$
\delta S = \mathsf{EL} + \check\pi_{0}^* \check{\alpha}; \qquad \check{\omega}=\delta \check{\alpha}
$$
where $\mathsf{EL}=\int_M \mathsf{el}$, and $\check{\alpha}$ is a $1$-form on $\cFclp_{\Sigma}$ obtained by integration of $d\gamma$ (see e.g.~\cite[Theorem 3]{Zuckerman:1986} or \cite{KiT1979}).

When the (pre-)symplectic reduction by the kernel of $\check{\omega}^\flat$ is smooth\footnote{This is not guaranteed. It turns out to be the case for most theories of interest including general relativity \cite[Proposition 3.2]{CattaneoSchiavinaEH}.}, it yields the surjective submersion
$$
  \pi_{\text{red}}\colon \cFclp_{\Sigma} 
  \longrightarrow \Fclp_{\Sigma}
$$
and, assuming that $\check\alpha$ is basic with respect to this fibration, i.e.\ $\check\alpha=\pi_{\text{red}}^*\alpha^\partial$ for some $\alpha^\partial \in \Omega^1_{\text{loc}}(\Fclp_{\partial M})$, we have the improved bulk-boundary relation
\begin{equation}\label{e:clBBrel}
    \delta S = \mathsf{EL} + \pi_{0}^* {\alpha}^\partial,
\end{equation}
where $\pi_{0}=\pi_{\text{red}}\circ \check\pi_{0} \colon \Fcl \to \Fclp_{\partial M}$.

The space $(\Fclp_{\Sigma},\omega^\partial = \delta \alpha^\partial)$ is an exact symplectic manifold, modeled on the space of sections of some (induced) fibre bundle on $\Sigma_{0}$. 

\begin{remark}
Observe that since $M$ has the structure of a finite cylinder $\Sigma\times [0,1]$, we can extend the procedure above to define the space of pre-boundary fields $\cFclp_{\partial M}$ as follows. Since $\partial M = \Sigma_0\sqcup\Sigma_1$ splits into incoming and outgoing boundary, the outcome of the restriction procedure is diffeomorphic to two copies of $\cFclp_\Sigma$. Performing presymplectic reduction for $\cFclp_{\partial M}$ returns two copies of the presymplectic reduction of $\cFclp_{\Sigma}$, with the sign of the symplectic structure on the second one reversed; i.e.\ we have $\pi_M\colon \Fcl \to \Fclp_{\partial M}\simeq\overline{\Fclp_\Sigma}\times{\Fclp_\Sigma}$. 

The space $\Fclp_\Sigma$ is called the geometric phase space,\footnote{We follow here \cite{KiT1979} for the general construction, although the terminology might differ.} or the space of initial data of the system. It is different from the \emph{physical}, or \emph{reduced} phase space, which will be defined in Section \ref{s:BFVth} as the coisotropic reduction of a submanifold defined by ``constraints''. 
\end{remark}

\begin{example}\label{ex:GRclass}
Looking ahead to Section \ref{s:gr}, we begin here our discussion of the example of general relativity. In that case, the space of fields is $\calF=\mathcal{Lor}_\Sigma(M)$: the space of lorentzian metrics on $M=\Sigma\times [0,1]$ whose restriction to $\partial M =\Sigma_0 \bigsqcup\Sigma_1$ is positive definite. The space of pre-boundary fields (for each boundary component) is  
$$
\cFclp_\Sigma \simeq T\mathbfcal{R}(\Sigma) \times C^\infty(\Sigma) \times {\calX}(\Sigma),
$$ 
where $\mathbfcal{R}(\Sigma)$ is  the space of riemannian metrics on $\Sigma$. The manifold $\cFclp_\Sigma$ is parametrized by a riemannian metric $h$ on $\Sigma$, the boundary value of its \emph{normal jet} $\dot{h}$, as well as a function $\eta\in C^\infty(\Sigma)$  and a vector field $\beta\in{\calX}(\Sigma)$, usually denoted by \emph{lapse} and \emph{shift}. It can then be shown that the form $\check{\omega}$, constructed as above, is pre-symplectic, and that the space of boundary fields is given by $\Fclp_\Sigma = T^*\mathbfcal{R}(\Sigma)$. For the explicit details of this calculation, see \cite[Section B]{CattaneoSchiavinaEH}. 
\end{example}

We will focus on variational problems for which the initial-value problem is well-posed. However, generally, not all points in $\Fclp_\Sigma$ can be extended to local solutions of the Euler--Lagrange equations in some neighborhood of $\Sigma$ in $M$. The subset $C\subset\Fclp_\Sigma$ consists of those  initial data that can be extended to a solution for some thin, finite cylinder $\Sigma\times[0,\epsilon]$. Typically, a good ansatz for $C$, which plays the role of Cauchy data for an initial value problem, is given through a vanishing ideal $\mathcal{I}_C$, derived from the Euler--Lagrange $1$-form $\mathsf{EL}$. Indeed, on a cylinder it is frequently possible to split equations of motion into \emph{evolution equations} and auxiliary relations among the fields.\footnote{One can always look at the equations of motion in a tubular neighborhood of some ``initial'' surface, and thus obtain a reasonable ansatz for $C$ by splitting the equations of motion there. This makes sense because $C$  contains information only on  local extendability of initial data.} The equations of motion that are not evolution equations are then seen as constraint functions $\{\phi_i\in C^\infty(\Fclp_\Sigma)\}$ for the configurations in $\Fclp_\Sigma$, so that the vanishing ideal $\mathcal{I}_C$ is generated by the functions $\phi_i$. The subset of Cauchy data often turns out to be coisotropic, which may be expressed by the statement that the constraints are \emph{first class}, i.e.~$\mathcal{I}_C$ is a Poisson subalgebra.  

In standard gauge theories, the coisotropic property of $C$ can be related to the existence of a hamiltonian action of the (gauge) symmetry group on initial data, for which $C$ is the preimage of zero under the associated (equivariant) momentum map. More generally, we will argue that this can be concluded from the existence of a BV-BFV field theory (see Def.\ \ref{def:BVBFV}). See also Remark \ref{rmk:BFVcoisotropicexplanation1}.

In Section \ref{s:Outlook}, we will discuss other possible approaches to the problem that are promising, but outside of the scope of the present paper.

\subsection{BFV data and coisotropic submanifolds}\label{s:BFVth}
Consider the submanifold $C\subset \Fclp$ of  ``Cauchy data'' discussed in Section \ref{s:LFTBoundary}, and assume that it is coisotropic.
The reduced phase space of the system is defined as the reduction $\underline{C}$, i.e.\ the leaf space of the characteristic foliation integrating the kernel $TC^\omega$ of the induced 2-form. Typically, $\underline{C}$ is not smooth (and in the example of GR, $C$ is not smooth either \cite{ArmsMarsdenMoncrief:1982}), and we resort to a cohomological replacement.

The BFV construction provides a cohomological resolution of the reduction $\underline{C}$, namely a complex $(C_{\sfBFV}^\bullet,Q_{\sfBFV})$ that is both positively and negatively graded\footnote{For $C_{\sfBFV}^\bullet$ to be an actual resolution one should assume that its cohomology groups in negative degree vanish. In field theory, a less restrictive requirement is that they be finite dimensional.}, such that 
\begin{equation}
    H_{\sfBFV}^0 \simeq C^\infty(\underline{C}) \simeq (C^\infty(M)\slash \mathcal{I}_C)^{\mathcal{I}_C},
\end{equation}
where $\mathcal{I}_C\subset C^\infty(M)$ is the vanishing ideal of $C$,  and the superscript $\mathcal{I}_C$ means the Poisson bracket commutant of $\mathcal{I}_C$ (see \cite{BatalinVilkovisky77,batalin1983generalized,StasheffConstraints88,StasheffHomRedPoiss,SchaetzBFV}). This construction is often taken as a definition for the space of functions on $C$ that are invariant with respect to the characteristic foliation. The cohomology of the BFV complex is isomorphic to the Lie algebroid cohomology associated to the coisotropic submanifold \cite[Corollary 3]{SchaetzBFV}.  For a practical procedure to construct BFV data given a coisotropic submanifold we refer to \cite{StasheffConstraints88} and \cite{SchaetzBFV}. Then, given an LFT on a manifold with boundary, we can construct a BFV complex to resolve $\underline{C}$, where $C$ is the submanifold of constraints/Cauchy data defined in Section \ref{s:LFTBoundary}. 

\begin{remark}\label{rmk:BFVcoisotropicexplanation1}
In what follows, instead of constructing the BFV complex starting from the knowledge of coisotropic Cauchy data, we will reverse the logic. We will induce the BFV data as structural boundary information from ``bulk'' cohomological data for an LFT on a manifold with boundary, as outlined in Section \ref{s:BVBFV}. This allows us to conclude that the space of initial data $\Fclp_\Sigma$ is endowed with a coisotropic submanifold of Cauchy data, the reduction of which is resolved by the BFV complex we induced from bulk data. See Corollary~\ref{thm:BFVcoisotropicexplanation2} for the particular example of general relativity.
\end{remark}

To this aim, it is convenient to introduce the following notion, which is close to that of \cite[Definition 2.2]{CMR18}:
\begin{definition}
A BFV theory is a quadruple  $(\calF_{\sfBFV},\Omega_{\sfBFV},S_{\sfBFV},Q_{\sfBFV})$ where:
\begin{itemize}
    \item $\calF_{\sfBFV}$ is a graded manifold called the space of BFV fields,
    \item $\Omega_{\sfBFV}$ is a weak $(0)$-symplectic structure, 
    \item $S_{\sfBFV}\in \Omega^0_{\text{loc}}[1](\calF_{\sfBFV})$ is a local function of degree $1$,
    \item $Q_{\sfBFV}\in{\calX}_{\text{loc}}[1](\calF_{\sfBFV})$ is a local vector field which is cohomological in the sense that $[Q_{\sfBFV},Q_{\sfBFV}]=0$, called the BFV operator. 
\end{itemize} 
such that 
$\iota_{Q_{\sfBFV}} \Omega_{\sfBFV} = \delta S_{\sfBFV}$, i.e.\ $Q_{\sfBFV}$ is a hamiltonian vector field with hamiltonian function $S_{\sfBFV}$. The BFV theory is said to be \emph{exact} if $\Omega_{\sfBFV}$ is an exact symplectic form.

The BFV complex is obtained as the space of smooth functions over the graded manifold: $C_{\sfBFV}^\bullet\doteq C^\infty(\mathcal{F}_{\sfBFV})$, endowed with the cohomological vector field $Q_{\sfBFV}$, seen as a differential.
\end{definition}

Given a BFV theory we can extract the information of a coisotropic submanifold\footnote{
In fact, there is a correspondence between coisotropic submanifolds of a Poisson manifold and ``BFV charges'' $S_{\sfBFV}$ in a neighborhood of such a submanifold. This is presented in \cite[Proposition 2.5]{SchaeInv}, which is based on the constructions of \cite{StasheffHomRedPoiss}. What we give here is an explicit argument, valid for our case of interest.} ${C}\subset \mathcal{M}\doteq\mathrm{Body}(\calF_{\sfBFV})$, the body of the graded BFV space of fields, by defining ${C}$ as the vanishing locus of the coefficients of the degree $1$-homogeneous part in the degree $1$ variables (or the degree $0$ homogeneous part in the variables of negative degree; i.e. the part independent of those variables, see below) of $S_{\sfBFV}$.

\begin{proposition}\label{prop:coisofromBFV}
Given a BFV theory, the ideal generated by $\frac{\partial S_{\sfBFV}}{\partial c}\vert_{c=0}$, where $c$ denotes a coordinate in degree $1$, is a Poisson subalgebra of the Poisson algebra defined by $\Omega_{\sfBFV}$.
\end{proposition}
\begin{proof}

Assume, for simplicity, that $\mathcal{F}_{\sfBFV}=\mathcal{M}\times {V}[1]\oplus {V}^*[-1]$, where $\mathcal{M}$ is a symplectic manifold and $V$ is a vector space. Denote by $(m,b,c)$ the variables on $\mathcal{M},V^*[-1]$ and $V[1]$ respectively, where $b$ and $c$ are dual to one another.

As a function of (total) degree $1$, we can decompose $S_{\sfBFV}$ with respect to its components of a given degree in $b$ as follows:
\[
    S_{\sfBFV}(x,b,c) = S^{(0)}(x,c) + S^{(1)}(x,b,c) + \dots 
\]
where the sum is finite and $S^{(i)}$ is homogeneous of degree $i$ in $b$ and of degree $i+1$ in $c$, so that $S^{(0)}$ is independent of $b$ and linear in $c$, while $S^{(1)}$ is linear in $b$ and quadratic in $c$.

Denoting by $\{\cdot,\cdot\}$ the Poisson bracket on $\calF_{\sfBFV}$ and by $\{\cdot,\cdot\}_0$ the Poisson bracket on $\mathcal{M}$, we have that the master equation $\{S_{\sfBFV},S_{\sfBFV}\}=0$, decomposes into equations for the parts of the bracket $\{S_{\sfBFV},S_{\sfBFV}\}$ of homogeneous degree in $b$.
The vanishing of the $b$-independent part of the bracket (which is homogeneous of degree 2 in $c$), reads:
\[
    \{S^{(0)},S^{(0)}\}_0 = 2\frac{\partial S^{(1)}}{\partial b}  \frac{\partial S^{(0)}}{\partial c} = 2\frac{\partial S^{(1)}}{\partial b} S^{(0)}
\]
by virtue of the linearity of $S^{(0)}$.

Hence, the $b$-independent part of the classical master equation implies that the ideal generated by $S^{(0)}$ is a Poisson subalgebra; this means that the clean zero locus\footnote{By the clean zero locus we mean the smooth points of the zero locus where a vector is tangent to $S^{(0)}=0$ if and only if it annihilates the differential
of $S^{(0)}$.} of $S^{(0)}$, which is by definition the constraint set, is coisotropic.
\end{proof}

The resolution of the reduction $\underline{{C}}$ is then given by the associated BFV complex. Furthermore, if the BFV theory came from the resolution of the Cauchy data for an LFT, the data extracted from this procedure will be equivalent to the boundary data associated to the field theory\footnote{This means that one might end up with quasi-isomorphic complexes. Observe furthermore that the form of the generators of $\mathcal{I}_C$ is arbitrary.} (as constructed in Section \ref{s:LFTBoundary}).

The BFV algebra associated to the reduction $\underline{C}$ is related to the notion of \emph{resolution by homotopy Lie--Rinehart algebras}\footnote{In \cite{Huebschmann17} the notion of strong homotopy Lie--Rinehart algebra is discussed. This relates to work of Kjeseth on homotopy Lie-Rinehart pairs \cite{KjesethHomotopyLRP}.} of the Lie--Rinehart algebra $\left(C^\infty(\calF^\partial_\Sigma)/\mathcal{I}_{C}, \mathcal{I}_{C}/\mathcal{I}_{C}^2\right)$ (see \cite[Theorem 4.1]{kjeseth2001homotopy}). Oh and Park defined another possible resolution of the Poisson algebra $C^\infty(\underline{C},\{,\}_{\underline{C}})$ associated to a coisotropic submanifold of a Poisson manifold, by means of a strong-homotopy Lie algebroid in \cite{OhParkSHLA}. The relation between these two constructions has been discussed in \cite[Theorem 5]{SchaetzBFV}.

\subsection{BV data for lagrangian field theories}\label{s:BVth}
Lagrangian field theories often enjoy local (or ``gauge'') symmetries. In what follows, symmetries will be specified by a Lie algebroid 
\begin{equation*}
    \pi_{\mathsf{A}}\colon\mathsf{A}\to \Fcl;  \qquad \rho\colon \mathsf{A} \to T\Fcl,
\end{equation*}
whose sections are mapped by the anchor $\rho$ to (local) vector fields that preserve the action functional.\footnote{More generally one can look at vector bundles $E\to \Fcl$, anchored by $\rho:E\to T\Fcl$ such that the image of $\rho$ is involutive (as a distribution) at least on $EL$.} The prototypical example is given by a Lie algebra action on $\Fcl$, where $\mathsf{A}$ is  the action Lie algebroid. Denoting by $\mathrm{Im}(\rho)=D$ the distribution of local symmetries, the space of physically inequivalent configurations is given by the \emph{moduli space of solutions} $EL/D$, where the quotient denotes that two solutions are identified if they can be related by a path of solutions integrating  a path of infinitesimal symmetries. This space is often hard to describe, and we look for a cohomological replacement.

Such a replacement can be found within the Batalin--Vilkovisky formalism. Given the data $(\Fcl,\Scl,\mathsf{A})$ for an LFT with local symmetries, one constructs a cochain complex in both positive and negative degrees\footnote{Coordinate functions in positive degree are often called ``ghost fields'' in physics terminology, while those in negative degree are called ``antifields''  or ``antighosts'', depending on their degree.}, whose cohomology in degree zero is taken as a replacement for the space of functions over $EL$ that are invariant under the ``action'' of the symmetries. Indeed, the BV complex is a combination of the Koszul--Tate complex, which resolves the quotient $C^\infty(\Fcl)/\mathcal{I}_{EL}$ by the vanishing ideal $\mathcal{I}_{EL}$ of $EL$, and the Chevalley--Eilenberg Lie algebroid complex to describe invariant functions \cite{stasheff1998secret,Mnev2017}.

\begin{definition}[\cite{CattaneoMnevReshetikhin2014}]
A BV theory is specified by a 4-tuple
\begin{equation*}
    (\calFBV,\OmegaBV,\SBV,Q_{\sfBV})
\end{equation*}
where 
\begin{itemize}
    \item $\calFBV$ is a graded manifold called the space of BV fields,
    \item $\OmegaBV$ is a (weak) $(-1)$-symplectic structure on $\calFBV$,
    \item $\SBV$ is a local function of degree zero called the BV action,
    \item $Q_{\sfBV}\in{\calX}_{\text{loc}}[1](\calFBV)$ is a local vector field, called the BV operator, which is cohomological in the sense that $[Q_{\sfBV},Q_{\sfBV}]=0$, 
\end{itemize} 
such that
\begin{equation}\label{e:BVHameq}
    \iota_{Q_{\sfBV}}\OmegaBV = \delta \SBV,
\end{equation}
which implies that $\SBV$ satisfies the \emph{classical master equation} 
\begin{equation*}
    \{\SBV,\SBV\} = 0
    \,.
\end{equation*}
\end{definition}

\begin{remark}
Recall that, on weak-symplectic manifolds, a Poisson bracket can be defined only on functions that admit a hamiltonian vector field. On such hamiltonian functions we can define a Poisson bracket by means of their hamiltonian vector fields: $\{f,g\}\doteq \iota_{X_f}\iota_{X_g}\Omega_{\sfBV}$.
\end{remark}

In order to make sure that a BV theory represents the classical data $(\Fcl,\Scl,\mathsf{A})$, one requires that $\Fcl$ be the body of the graded manifold $\calFBV$ and that $\SBV\vert_{\Fcl}=\Scl$. 
Then, the BV complex is given by the functions $C^\bullet_{\sfBV} = C^\infty(\calFBV)$, with differential $Q_{\sfBV}$ seen as a derivation of the algebra of smooth functions. Its cohomology in degree zero is
$$
    H^\bullet_{\sfBV} = (C^\infty(\Fcl)/\mathcal{I}_{EL})^D \simeq C^\infty(EL)^D
$$
the space of functions on the critical locus which are invariant with respect to the action of the symmetry distribution $D$, which is the the image of the anchor of the Lie algebroid $\mathsf{A}$ over $\mathcal{F}$.

In order to construct a BV theory from the data of a LFT $(\Fcl,\Scl,\mathsf{A}),$ one looks at $\calF_{\sfBV}=T^*[-1]\mathsf{A}[1]$ and the BV action function $S_{\sfBV} = \pi_{\mathsf{A}}^*\Scl + \tilde{Q}_{CE},$ where $\tilde{Q}_{CE}$ is the Chevalley--Eilenberg differential associated with $\mathsf{A}$, seen as a function in $C^\infty(T^*[-1]\mathsf{A}[1])$. The BV differential is then given by the local, cohomological, vector field
\begin{equation*}
    Q_{\sfBV} = \hat{Q}_{CE} + Q_{K},
\end{equation*}
where $\hat{Q}_{CE}$ is the cotangent lift of the Chevalley--Eilenberg differential, seen as a cohomological vector field on $T^*[-1]\mathsf{A}[1]$, while $Q_K$ is the Koszul differential\footnote{Up to boundary terms, $Q_K$ coincides with the hamiltonian vector field $X_{\pi_{\mathsf{A}}^*\Scl}$ of the action function pulled back to the shifted cotangent bundle.} for the vanishing ideal  $\mathcal{I}_{EL}$. Then $(Q_{\sfBV},S_{\sfBV})$ is a hamiltonian pair.

\begin{remark}
Crucial to our discussion is the observation that the  construction above works well whenever the field theory is defined on a compact manifold without boundary, or else for fields with appropriate asymptotic or boundary conditions. On manifolds with boundary, instead, Equation \eqref{e:BVHameq} is spoiled by boundary terms. As we will see, this is a first step towards connecting bulk and boundary data.
\end{remark}

\subsection{Connecting bulk to boundary: the BV-BFV framework}\label{s:BVBFV}
If we are given the data of a BV theory in the interior of a manifold with boundary, this will generally fail to extend to the closure as a BV theory, because boundary terms will spoil Equation \eqref{e:BVHameq}. In this situation, the data can at times be complemented with a BFV theory to give rise to what we call a BV-BFV theory. This will turn out to be a good model for a LFT on a manifold with boundary.
\begin{definition}[\cite{CattaneoMnevReshetikhin2014}]\label{def:BVBFV}
We define a BV-BFV theory over the exact BFV theory $(\calF_{\sfBFV},\Omega_{\sfBFV}=\delta \alpha_{\sfBFV},S_{\sfBFV},Q_{\sfBFV})$ to be the quadruple $(\calFBV,\OmegaBV,\SBV,Q_{\sfBV})$, where\footnote{We allow ourselves to abuse notation by using the same symbols we used for BV theories.} 
\begin{itemize}
    \item $\calFBV$ is a graded manifold,
    \item $\OmegaBV$ is a (weakly) $(-1)$-symplectic structure on $\calFBV$,
    \item $\SBV$ is a local function of degree zero called the BV action,
    \item $Q_{\sfBV}\in{\calX}_{\text{loc}}[1](\calFBV)$ is a local vector field which is cohomological in the sense that $[Q_{\sfBV},Q_{\sfBV}]=0$, called the BV operator,
\end{itemize}
together with a surjective submersion $\pi_{\sfBV}\colon \calFBV\to\calF_{\sfBFV}$ such that:
\begin{itemize}
    \item The vector field $Q_{\sfBV}$ is projectable along the surjective submersion $\pi_{\sfBV}$ and $Q_{\sfBFV}$ is its projection, i.e.\ $\pi_{\sfBV}\circ Q_{\sfBV} = Q_{\sfBFV} \circ \pi_{\sfBV}$, 
    \item We have the compatibility:
\begin{equation}\label{e:bulkboundaryBV}
    \delta \SBV = \iota_{Q_{\sfBV}}\OmegaBV + \pi_{\sfBV}^*{\alpha}_{\sfBFV}.
\end{equation}
\end{itemize}
\end{definition}

\begin{remark}
A consequence of this is that the BV action $\SBV$ now satisfies the \emph{modified classical master equation}:
\begin{equation*}
    \{\SBV,\SBV\}\doteq \iota_{Q_{\sfBV}}\iota_{Q_{\sfBV}}\OmegaBV = \pi_{\sfBV}^* S_{\sfBFV},
\end{equation*}
which tells us that the BFV action controls the anomaly in the classical master equation induced by the presence of a boundary.
\end{remark}

\begin{remark}[Inducing a BV-BFV theory from a BV theory]\label{rem:induction}
On a manifold with boundary we can interpret the obstruction to $(Q_{\sfBV},\SBV)$ being a hamiltonian pair as a $1$-form $\check{\alpha}_{\sfBV}$ on $\cFclp_{\sfBV}$, the space of germs of sections in a thin neighborhood of the boundary (compare with Equation \eqref{e:clBBrel}); i.e.\ we have: 
\begin{equation}\label{e:bulkpreboundaryBV}
    \delta \SBV = \iota_{Q_{\sfBV}}\OmegaBV + \check\pi_{\sfBV}^*\check{\alpha}_{\sfBFV}.
\end{equation}
The induced $2$-form $\check{\Omega}_{\sfBFV}=\delta \check{\alpha}_{\sfBFV}$ is degenerate and, analogously to the procedure outlined\footnote{This means that we can replace $\check{\mathcal{F}}_{\partial M}^\partial$ with $\check{\mathcal{F}}_{\sfBFV}^\partial$ and $\check{\alpha}$ with $\check{\alpha}_{\sfBFV}$ everywhere.} in Section \ref{s:LFTBoundary}, if its kernel is regular, one constructs the space of BFV fields as 
$$
\calF_{\sfBFV} \doteq \cFclp_{\sfBV}\slash \mathrm{ker}(\check{\Omega}_{\sfBV}^\sharp).
$$
If this quotient is smooth, one can assign a BFV theory to the boundary submanifold and obtain a BV-BFV theory. Moreover, if $\check{\alpha}_{\sfBFV}=\pi_{\text{red}}^*\alpha_{\sfBFV}$ is basic, the induced BFV theory is exact, and Equation \eqref{e:bulkpreboundaryBV} becomes Equation \eqref{e:bulkboundaryBV}. The BFV theory obtained in this way is symplectomorphic to the one obtained from the coisotropic submanifold $C\subset\Fclp_\Sigma$ induced from the lagrangian field theory. In other words we can summarize:
\begin{equation*}
    \xymatrix{
    (\Fcl,\Scl,\mathsf{A}) \ar[d]^{\pi} \ar@{~>}[r]^-{\sfBV} & (\calFBV,\OmegaBV,\SBV,Q_{\sfBV}) \ar[d]^{\pi_{\sfBV}} \\
    (\Fclp_\Sigma,C) \ar@{~>}[r]^-{\sfBFV} & (\calF_{\sfBFV},\Omega_{\sfBFV},S_{\sfBFV},Q_{\sfBFV})
    }
\end{equation*}
where the undulating arrows indicate the BV and BFV constructions.
\end{remark}

\section{General relativity}
\label{s:gr}

This section is mostly a review of \cite{CattaneoSchiavinaEH,ScTH}. Corollary \ref{thm:BFVcoisotropicexplanation2} clarifies how one can use the BFV structure induced from the bulk BV theory to prove that the vanishing locus of the energy and momentum constraints (Definition \ref{def:GRconstraints}) forms a \emph{coisotrope}. Moreover, in Section \ref{s:Linftyanalysis} we expand on previous work by making the $L_\infty$ structure underlying the BFV data explicit.

In what follows we will work with cylindrical space-time manifolds with boundary; e.g.\ $M=\Sigma\times [0,1]$. We wish to have both boundary components $\Sigma$ be spacelike (hence Cauchy) surfaces, i.e., the induced metrics there are riemannian.  In addition, for simplicity, we assume $\Sigma$ to be compact and without boundary. Consider the space of such lorentzian metrics on $M$, denoted $\mathcal{Lor}_\Sigma(M)$.

The Lie algebra ${\calX}(M)$ of vector fields on $M$ (with no restriction on the boundary values) acts on $\mathcal{Lor}_\Sigma(M)$ by Lie derivatives, which we describe by an action algebroid $\pi_{\text{diff}}\colon \mathsf{A}_{\mathrm{diff}} \to \mathcal{Lor}_\Sigma(M)$ with fibres given by ${\calX}(M)$ and anchor map 
\begin{equation*}
\rho_{\text{diff}}\colon (g,X) \longmapsto (g,L_X g)
\end{equation*}
where $(g,L_Xg)$ is a tangent vector in $T_g\mathcal{Lor}_\Sigma(M)$. The bracket on constant sections of $\mathsf{A}_{\text{diff}}$ is induced by the bracket of vector fields on $M$: 
\begin{equation*}
[(g,X),(g,Y)]_{\mathsf{A}_{\text{diff}}} = (g,[X,Y]_{{\calX}(M)}),
\end{equation*} extended to general  sections by the Leibniz rule.

\subsection{BV structure for general relativity}\label{s:BVGR}
The formulation of general relativity that we will consider in this paper relies on the Einstein--Hilbert (EH) lagrangian field theory, specified by $(\mathcal{Lor}_\Sigma(M),S^{EH},\mathsf{A}_{\text{diff}})$, where $S^{EH}$ is the action function
\begin{equation}\label{e:EHclassaction}
    S^{EH}=\intl_M R(g) \Vol_g.
\end{equation}

To construct the (standard\footnote{A non-standard, equivalent, BV formulation for Einstein--Hilbert theory was presented in \cite[Theorem 37 and Remark 40]{CanepaCattaneoSchiavina2021}.}) BV theory for EH gravity we apply the general construction outlined in Section \ref{s:BVth} to the Einstein--Hilbert theory $(\mathcal{Lor}_\Sigma(M),S^{EH},\mathsf{A}_{\text{diff}})$. The space of BV fields will then be
$$
    \calF_{\sfBV}^{EH}=T^*[-1]\mathsf{A}_{\text{diff}}[1]=T^*[-1]\left(\mathcal{Lor}_\Sigma(M)\times {\calX}[1](M)\right).
$$

If we let $\xi\in{\calX}[1](M)$ denote an odd coordinate (of degree 1), we can write down the Chevalley--Eilenberg part of $Q_{\sfBV}^{EH}$ as: 
\begin{align}\label{e:Qdiff}
Q_{CE}^{EH} g = L_{\xi{}} g\qquad Q_{CE}^{EH} \xi{} = \frac12 [\xi{},\xi{}].
\end{align}

We can think of the vector field $Q_{CE}^{EH}$ either as a section of the tangent bundle of $\mathcal{Lor}_\Sigma(M)\times {\calX}[1](M)$ or as a derivation of the algebra of functions of $g$ and $\xi$. From the first viewpoint, we think of the space $\mathcal{Lor}_\Sigma(M)$ as an open subset of the vector space $S^2(T^*M)$  of all symmetric covariant 2-tensors on $M$, so that all of its tangent spaces may be identified with $S^2(T^*M)$. The first equation in \eqref{e:Qdiff} can then be taken as a formula for the component of the vector field $Q_{CE}^{EH}$ in the $\mathcal{Lor}_\Sigma(M)$ direction as a function of $(g,\xi)$.  Similarly for the second equation, which gives the component in the direction of the vector space ${\calX}[1](M)$.  

From the second viewpoint, we think of $g$ as standing for the vector-valued function on $\mathcal{Lor}_\Sigma(M)\times {\calX} [1](M)$ given by projection onto the first factor followed by the inclusion into $S^2(T^*M)$, in which case the first equation in \eqref{e:Qdiff} gives the result of applying the derivation $Q_{CE}^{EH}$ to this function. Similarly, the second equation gives the result of applying the derivation to the projection onto the second factor.  Note that, since $\xi$ is a variable of degree 1, 
the bilinear expression $[\xi{},\xi{}]$ is to be understood as an element of the space $$\mathrm{Sym}^2({\calX}[1](M),{\calX}[1](M))$$  of graded-symmetric bilinear maps from ${\calX}[1](M)$ to itself, namely the element of $$\mathrm{Hom}(\wedge^2{\calX}(M),{\calX}(M))$$  with input  a pair  $(X,Y)$ of vector fields and output $[X,Y]$.  Note that this is a (quadratic) function on $ {\calX} [1](M)$ with values in the vector space $ {\calX} [1](M)$.

We can construct a BV action as prescribed in Section \ref{s:BVth}:
\begin{equation*}
    S_{\sfBV}^{EH}=\pi_{\mathsf{A}_{\text{diff}}}^* S^{EH} + \tilde{Q}_{CE}^{EH} = \intl_M R(g) \Vol_g + \langle g^\dag, L_\xi g\rangle + \iota_{[\xi,\xi]}\xi^\dag,
\end{equation*}
where $g^\dag\in S^2(TM)\otimes \mathrm{Dens}(M)$ and $\xi^\dag\in\Omega^1(M)\otimes \mathrm{Dens}(M)$ denote fields in the cotangent fibres of $\calF_{\sfBV}$. Up to boundary terms, the BV operator $Q_{\sfBV}^{EH}=\hat{Q}_{CE}^{EH}+Q_K^{EH}$ is the hamiltonian vector field of $S_{\sfBV}^{EH}$.

\subsection{BFV structure for general relativity}\label{s:BFVGR}
In this section we will present the BFV structure for Einstein--Hilbert theory. This has been obtained from the BV data constructed above by one of the authors together with Cattaneo in \cite{CattaneoSchiavinaEH}. 

Let us denote the space of riemannian metrics on $\Sigma$ by  $\mathbfcal{R}(\Sigma)$. We will denote by $\mathcal{D}(\Sigma)$ the space of densities on $\Sigma$, i.e.\ the sections of the density bundle $\mathrm{Dens}(\Sigma)\to \Sigma$.
The tangent bundle $T\mathbfcal{R}(\Sigma)$ may be identified with the trivial bundle whose fibre $T_h \mathbfcal{R}(\Sigma)$ over the metric $h$ is the vector space  of (smooth) sections of the bundle $S^2(T^*\Sigma)$ of symmetric covariant 2-tensors. 
By the cotangent bundle $T^*\mathbfcal{R}(\Sigma)$ we will mean the trivial bundle whose fibre over $h$ is the vector space of smooth sections $\Pi$ of $S^2(T\Sigma)\otimes\mathrm{Dens}(\Sigma)$.

\subsection*{Notation}\label{par:notation}
A metric $h$ on $\Sigma$ induces ``musical'' isomorphisms $\flat$ and $\sharp$ between $T\Sigma$ and $T^*\Sigma$. These lift to isomorphisms between their spaces of sections.  With the aid of the natural density attached to each metric $h$, we obtain isomorphisms between the tangent and cotangent spaces of $\mathbfcal{R}(\Sigma)$.  This induces, in turn, a map between bilinear forms and endomorphisms, so that 
$$\Pi^\flat\colon \Gamma(\Sigma,T\Sigma) \to \Gamma(\Sigma,T\Sigma)\otimes \mathcal{D}(\Sigma)$$ 
is a density-valued endomorphism of the tangent space for each cotangent vector $\Pi$.   We denote its square by $\Pi^2\doteq\Pi^\flat\circ\Pi^\flat$, which we may think of as a 2-density-valued endomorphism.
Occasionally it will be useful to raise/lower indices twice, so we use the notation $\Pi^{\flat\flat}$ to denote the section of $S^2(T^*\Sigma)\otimes \mathrm{Dens}(\Sigma)$ obtained from $\Pi$ by precomposition with $h\otimes h$, so that in a coordinate chart we have (summing over repeated indices)
\begin{align*}
[\Pi^{\flat\flat}] &= h_{ac}\Pi^{cd}h_{db}dx^a\odot dx^b = \Vol_h\otimes h_{ac}\pi^{cd}h_{db}dx^a\odot dx^b ,\\
[\Pi^2]^\sharp & =\Pi^{ab}h_{bc} \Pi^{cd}\partial_a \odot\partial_d = \Vol_h \otimes \Vol_h \otimes \pi^{ab}h_{bc}\pi^{cd}\partial_a\odot\partial_d 
\end{align*}
for some symmetric tensor $\pi\in \Gamma(\Sigma,S^2(T\Sigma))$.

With a slight abuse of notation, we denote the contraction of an element of $S^2(T^*M)$ with a metric $h$ as a trace operation $\mathrm{Tr}_h\doteq\langle h, \cdot\rangle$, where the angular brackets $\langle\cdot ,\cdot \rangle$ will generically denote the pairings 
$$
\langle\cdot,\cdot\rangle \colon S^k(T\Sigma)\otimes\mathrm{Dens}(\Sigma)\times S^k(T^*\Sigma) \to \mathrm{Dens}(\Sigma).
$$  
We will then use the notation $\Ric(h)$ for the Ricci curvature and $R(h)=\mathrm{Tr}_h(\Ric(h))$ for the scalar curvature, while the Einstein tensor will be denoted by ${G}(h) = \Ric(h) - \frac12 h R(h)$. Finally, given a function $\phi\in C^\infty(\Sigma)$, and denoting by $\nabla$ the covariant derivative with respect to the Levi-Civita connection associated with $h$, we define
\begin{equation*}
{D}_h(\phi) \doteq -\nabla\nabla \phi + h\mathrm{Tr}_h[\nabla\nabla\phi] = -\frac12\Lie_{\Grad_h \phi} h + \frac12 h \mathrm{Tr}_h[\Lie_{\Grad_h \phi} h].
\end{equation*}
Observe that $\nabla\nabla\phi \in S^2(T^*\Sigma)$.
\bigskip

In general relativity, the variation of the Einstein--Hilbert action \eqref{e:EHclassaction} can be split in two parts, using the global space-time splitting on the cylinder. By choosing the appropriate parametrization of a metric on the cylinder, for example using the Arnowitt--Deser--Misner (ADM) decomposition \cite{ArnowittDeserMisnerRepub}, we divide Einstein's equations into evolution equations and \emph{constraints}. Since the constraints only depend on a metric $h$, its first jet $\dot{h}$ and the \emph{lapse} and \emph{shift} components of a 4-metric, they restrict to equations on $\cFclp_\Sigma$ (cf.~Example \ref{ex:GRclass}). These can be formulated as the vanishing of functions on $\cFclp_\Sigma$ that are basic with respect to the pre-symplectic reduction $\pi_{\text{red}}\colon \cFclp_\Sigma \to \Fclp_\Sigma$. Hence, the constraints are given as the vanishing locus of functions on $\Fclp_{\Sigma}$, as follows:

\begin{definition}\label{def:GRconstraints}
The \emph{energy and momentum constraints}
\begin{align*}
  \mathbb{H}_n : C^\infty(\Sigma)
  &\longrightarrow C^\infty(T^* \mathbfcal{R})
  \\
  \mathbb{H}_\partial : {\calX}(\Sigma)
  &\longrightarrow C^\infty(T^* \mathbfcal{R})
\end{align*}
are defined by 
\begin{subequations}
\begin{align*}
    \mathbb{H}_n(\phi) &= \intl_\Sigma H_n(\phi) \doteq \intl_\Sigma\left\{\frac{1}{\Vol_h}\left(\mathrm{Tr}_h[\Pi^2] - \frac{1}{d-1}[\mathrm{Tr}_h\Pi]^2\right) + \Vol_h R(h)\right\} \phi  \\
    \mathbb{H}_\partial(X) &= \intl_\Sigma H_\partial(X) \doteq \intl_\Sigma \langle \Pi, \Lie_Xh \rangle ,
\end{align*}
\end{subequations}
where $\Lie_X h$ denotes the Lie derivative of the metric $h$ w.r.t. the vector field $X$. We call $\mathbb{H}_n(\phi)$ and $\mathbb{H}_\partial(X)$ the constraint functions, $H_n(\phi)$ and $H_\partial(X)$ the \emph{constraint densities\footnote{Observe that $H_\partial\in\Omega^1(\Sigma)\otimes\mathcal{D}(\Sigma)$.}}, respectively, and denote their associated vanishing ideal by $\mathcal{I}_{EH}$. 
\end{definition}

\begin{remark}
The constraint densities $H_n(\phi)$ and $H_\partial(X)$ are linear in $\phi$, $X$ and local in $\phi$, $X$, $h$, and $\Pi$. The constraint functions $\mathbb{H}_n(\phi)$ and $\mathbb{H}_\partial(X)$ are still linear but, due to the integration, no longer local.
\end{remark}

\begin{theorem}[\cite{Katz,deWittQuantumGravityI}]
The Poisson brackets of the constraint functions are given by:
\begin{subequations}
\label{e:ConstraintBrackets}
\begin{align}
    \{\mathbb{H}_\partial(X),\mathbb{H}_\partial(Y)\} &= \mathbb{H}_\partial([X,Y])\\
    \{\mathbb{H}_\partial(X),\mathbb{H}_n(\phi)\} &= \mathbb{H}_n(\Lie_X(\phi))\\
    \{\mathbb{H}_n(\phi),\mathbb{H}_n(\psi)\} &= \mathbb{H}_\partial(\phi\,\mathrm{grad}_h\psi - \psi\,\mathrm{grad}_h\phi). \label{e:bracketstructurefunctions}
\end{align}
\end{subequations}
\end{theorem}

\begin{remark}[Warning]
The right hand side of Eq.~\eqref{e:bracketstructurefunctions} is a customary abuse of notation, since the argument $\phi\,\mathrm{grad}_h\psi - \psi\,\mathrm{grad}_h\phi$ is not a vector field but a vector field valued function on $T^* \mathbfcal{R}$. This means that~\eqref{e:bracketstructurefunctions} is \emph{not} in the image of the momentum constraint of Definition~\ref{def:GRconstraints}.  The upshot is that the constraint functions are \emph{not} closed under the Poisson brackets, but generate a large Lie subalgebra.

It is shown in~\cite{BFW} that the Poisson brackets~\eqref{e:ConstraintBrackets} can be seen as the brackets of the constant sections of a Lie algebroid constructed on a much larger space than $T^*\mathbfcal{R}$.
One could try to view the pairs $(\phi,X)$ as constant sections of a trivial vector bundle over $T^*\mathbfcal{R}$ with fibre $C^\infty(\Sigma)\times {\calX}(\Sigma)$, and attempt to endow it with a Lie algebroid structure whose brackets of constant sections are given by Equations \eqref{e:ConstraintBrackets}.
Unfortunately this approach does not work, as we will see in Section \ref{s:Alternative}.
\end{remark}

\begin{remark}\label{rem:EHConstraintlocus}
The constraint set $C_{EH}\subset \Fclp_\Sigma = T^*\mathbfcal{R}(\Sigma)$ is  the vanishing locus $\mathcal{I}_{EH}$ of the energy and momentum constraint functions of Definition \ref{def:GRconstraints}, and it determines the field configurations $(h,\Pi)$ that can be extended to a solution of Einstein's equations in a thin neighborhood of $\Sigma$. It can be derived from the Einstein--Hilbert lagrangian field theory following the procedure outlined in Section \ref{s:BVBFV}.
\end{remark}

\begin{theorem}[\cite{ScTH,CattaneoSchiavinaEH}]\label{t:BVBFVEH}
The data $(\calF_{\sfBV}^{EH},\Omega_{\sfBV}^{EH},S_{\sfBV}^{EH},Q_{\sfBV}^{EH})$ define a BV-BFV theory over the exact BFV theory $(\calF_{\sfBFV}^{EH},\omega_{\sfBFV}^{EH},S_{\sfBFV}^{EH},Q_{\sfBFV}^{EH})$ defined by the following data. We introduce the shorthand notation 
\[
    \mathcal{V}(\Sigma) \doteq C^\infty(\Sigma) \oplus {\calX}(\Sigma).
\]
The graded $(0)$-symplectic space of BFV fields is given by:
\begin{subequations}\begin{align}
\mathcal{F}_{\sfBFV}^{EH} &\doteq 
T^*\left(\mathbfcal{R}(\Sigma)\times \mathcal{V}[1](\Sigma)\right), \label{e:BFVspaceEH}\\
\omega_{\sfBFV}^{EH}&\doteq\intl_\Sigma \bigl( \langle\delta h,\delta\Pi\rangle + \langle\delta\chi_\partial,\delta \xi^\partial \rangle + \delta \chi_n\, \delta\xi^n 
\bigr),
\end{align}\end{subequations}
where $\xi^n,\xi^\partial$ are respectively a degree-one function and a degree-one vector field on $\Sigma$, while 
\[
(\chi_n,\chi_\partial)\in \left(C^\infty[-1](\Sigma) \oplus\Omega^1[-1](\Sigma)\right)\otimes \mathcal{D}(\Sigma) \doteq \Phi[-1](\Sigma)
\]
denote variables in the cotangent fibre of $T^*\mathcal{V}(\Sigma)$, and $\langle\cdot,\cdot\rangle$ denotes the canonical fibrewise pairing. The BFV action function on $\mathcal{F}_{\sfBFV}$ reads:
\begin{subequations}
\begin{align*}
S_{\sfBFV}^{EH}=\mathbb{H}_n(\xi^n) + \mathbb{H}_\partial(\xi^\partial) + \intl_{\Sigma} \bigl(\chi_n \Lie_{\xi^\partial}\xi^n + \langle\chi_\partial,\xi^n\,\mathrm{grad}_h\xi^n\rangle + \frac12\langle\chi_\partial,[\xi^\partial,\xi^\partial]\rangle
\bigr),
\end{align*}
\end{subequations}
and the local, cohomological, vector field $Q_{\sfBFV}^{EH}$, hamiltonian\footnote{Note: the sign convention we adopt uses the total degree, summing internal grading and vertical form degree.  With this convention $\delta$ has degree $1$, $\iota_Q$ is an even derivation, and $\Lie_Q$ is odd.} with respect to $S_{\sfBFV}^{EH}$, is described by its action on fields (we drop all sub- and superscripts) as
\begin{subequations}\label{e:Qboundary}\begin{align}
Q(\xi^n)&=\Lie_{\xi^\partial}\xi^n \label{e:Qboundary-bracket1}\\
Q(\xi^\partial)&=\xi^n\, \mathrm{grad}_h\xi^n + \frac12[\xi^\partial,\xi^\partial]\label{e:Qboundary-bracket2}\\
Q(h)&=\widetilde{h}\xi^n -\Lie_{\xi^\partial}h = - 2K\xi^n - \Lie_{\xi^\partial} h \label{e:Qboundary-h}\\
Q(\Pi) &= \widetilde{\Pi}\xi^n +\Vol_h\left({G}^{\sharp\sharp}(h) \xi^n + {D}^{\sharp\sharp}_{h}(\xi^n)\right) - \Lie_{\xi^\partial}\Pi  -  (\chi_\partial\otimes_s d\xi^n)^{\sharp\sharp}\xi^n 
\label{e:Qboundary-Pi}\\
Q(\chi_\partial)&= {H}_\partial +  \Lie_{\xi^\partial}{\chi}_\partial - {\chi}_n d\xi^n \label{e:Qboundary-chiP}\\
Q(\chi_n)&={H}_n  +  \Lie_{\xi^\partial}\chi_n  - 2 \Lie_{\chi_\partial^\sharp}(\xi^n\Vol_h^{-\frac12})\Vol_h^{\frac12}, \label{e:Qboundary-chin}
\end{align}\end{subequations}
where
\begin{subequations}
\label{definitions}
\begin{align*}
\widetilde{\Pi}&\doteq\pard{}{h}\left(\frac{1}{\Vol_h}\left(\mathrm{Tr}_{h}[{\Pi}^2] - \frac{1}{d-1}\mathrm{Tr}_{h}[{\Pi}]^2\right)\right)\\\notag
&\phantom{:}=-\frac{1}{2\Vol_h}h^{-1}\left(\mathrm{Tr}_{h}[{\Pi}^2] - \frac{1}{d-1}\mathrm{Tr}_{h}[{\Pi}]^2\right) + \frac{2}{\Vol_h}\left([\Pi^2]^\sharp - \frac{1}{d-1}\Pi\mathrm{Tr}_h[\Pi]\right)\\
\widetilde{h}&\doteq\pard{}{{\Pi}}H_n = \frac{2}{\Vol_h}(\Pi^{\flat\flat} - \frac{h}{d-1}\mathrm{Tr}_{h}{\Pi}) = - 2 K.
\end{align*} \end{subequations}
and $\chi_\partial^\sharp$ is the densitized vector field obtained by applying the musical isomorphism to the $1$-form-density $\chi_\partial$. For an explicit expression of the BV-BFV map $\pi_{\sfBV}\colon \Fcl_{\sfBV}^{EH} \to \Fcl_{\sfBFV}^{EH}$ we refer to \cite[Theorem 3.6]{CattaneoSchiavinaEH}.
\end{theorem}

\begin{remark}\label{rem:reparametrizationanddensities}
Notice that in \cite{ScTH,CattaneoSchiavinaEH} a slightly different parametrization is used. In the original references the coordinate in the cotangent fibre is taken to be $\Pi_{\text{orig}}=\sqrt{h}\pi$ for some symmetric 2-tensor $\pi$, whereas  we use here the combination $\Pi=\sqrt{h}\pi\Vol = \pi \Vol_h = \Pi_{\text{orig}}\Vol$, where $\Vol$  is the Euclidean volume form. Notice that $\Lie_X (\Pi) = \Lie_X \Pi_{\text{orig}} + \Div(X)$, where $\Div(X)$ is the divergence of $X$ w.r.t. the Euclidean volume form. Given a $1$-form $\eta$ with associated vector field $\eta^\sharp$, and denoting the associated 1-form density and densitized vector field respectively by $\chi_\partial = \eta\otimes \Vol_h$ and $\chi_\partial^\sharp = \eta^\sharp\otimes \Vol_h$, we compute
\begin{equation}\label{e:explicithalfdensityformula}
    \Lie_{\eta^\sharp}(\xi^n\Vol_h^{-\frac12})\Vol_h^{\frac32} = \Lie_{\eta^\sharp}\xi^n \Vol_h + \frac12 \xi^n\Div_h(\eta^\sharp)\Vol_h =: \Lie_{\chi^\sharp}\xi^n + \frac12  \widetilde{\Div}_h(\chi_\partial^\sharp) \xi^n,
\end{equation}
where the second equality defines $\widetilde{\Div}_h(\chi^\sharp_\partial)$ as the density obtained from the 1-form density $\chi_\partial= \eta\otimes \Vol_h$ by 
$$
\widetilde{\Div}_h\colon \chi_\partial \mapsto \Div_h(\eta^\sharp)\otimes \Vol_h.
$$
Hence, recalling that $\chi^\sharp = \Vol_h \eta^\sharp,$ we have
\begin{equation}\label{e:densvfLie}
    \Lie_{\eta^\sharp}(\xi^n\Vol_h^{-\frac12})\Vol_h^{\frac32} \equiv \Lie_{\chi_\partial^\sharp}(\xi^n\Vol_h^{-\frac12}) \Vol_h^{\frac12}.
\end{equation}
Observe that Equation \eqref{e:densvfLie} should be taken as a definition of the r.h.s., i.e.\ of the Lie derivative of a (half-)density along a densitized vector field. The BFV data presented in \cite{CattaneoSchiavinaEH} uses the unfolded expression in Equation \eqref{e:explicithalfdensityformula}, with the original parametrization $(h,\Pi_{\text{orig}})$.
\end{remark}

\begin{corollary}[of Theorem \ref{t:BVBFVEH}]\label{thm:BFVcoisotropicexplanation2}
\label{c:main}
The vanishing ideal generated by the energy and momentum constraints is a Poisson subalgebra of the Poisson algebra on (hamiltonian functions on) $T^*\mathbfcal{R}(\Sigma)$.  Hence, the constraint set, given by the vanishing of the constraint functions of Definition \ref{def:GRconstraints}, is coisotropic.
\end{corollary}

\begin{proof}
The BV structure for GR given in Section \ref{s:BFVGR} induces a BFV structure on $\Sigma$ by means of Theorem \ref{t:BVBFVEH}. The terms in $S_{\sfBFV}^{EH}$ that are linear in the degree $1$ variables $\xi^n,\xi^\partial$ are $\mathbb{H}_n(\xi^n),\mathbb{H}_\partial(\xi^\partial)$. The BFV action $S_{\sfBFV}^{EH}$ satisfies the classical master equation
\[
    \left\{S_{\sfBFV}^{EH},S_{\sfBFV}^{EH}\right\}_{\omega_{\sfBFV}^{EH}}=0,
\]
where $\{\cdot,\cdot\}_{\omega_{\sfBFV}^{EH}}$ is the Poisson bracket induced (on hamiltonian functions) by $\omega^{EH}_{\sfBFV}$. Following Proposition \ref{prop:coisofromBFV}, we have that the vanishing of the part of $\left\{S_{\sfBFV}^{EH},S_{\sfBFV}^{EH}\right\}_{\omega_{\sfBFV}^{EH}}$ that is independent of $\chi_\partial,\chi_n$ implies equations \eqref{e:ConstraintBrackets}.
\end{proof}

\begin{remark}
The data of Theorem \ref{t:BVBFVEH} have been obtained from the bulk Einstein--Hilbert BV theory in \cite{ScTH,CattaneoSchiavinaEH}, and the fact that they satisfy the axioms of a BFV theory is a consequence of the BV-BFV induction procedure of \cite{CattaneoMnevReshetikhin2014} (see Remark \ref{rem:induction}). This, in turn, allowed us to conclude in Corollary \ref{c:main} that the constraint set is coisotropic, without prior knowledge of the brackets \eqref{e:ConstraintBrackets}.  
\end{remark}

\begin{remark}
Notice as well that the BFV data of Theorem \ref{t:BVBFVEH} yield a resolution of the coisotropic vanishing locus of the constraints of Definition \ref{def:GRconstraints}, and thus coincide with the BFV data one would obtain directly by means of the construction outlined in Section \ref{s:BFVth} (see \cite{FradkinVilkoviskyGrav}). However, since the bracket of constraints \eqref{e:bracketstructurefunctions} showcases a point-dependent ``structure constant'' (or better said, a structure function), it is a nontrivial statement that the correct BFV action can be taken to be linear in the degree $-1$ variables\footnote{These are commonly referred to as ``ghost momenta''.} $\chi_n,\chi_\partial$. A proof of this statement is however given by the BV-BFV induction procedure itself. A similar linearity statement holds for the coframe formulation of gravity, where the BFV data has instead been computed directly from the constraint functions in \cite{canepacattaneoschiavinaboundary}, due to the lack of a viable BV-BFV induction procedure in this case \cite{CS2017}. 
\end{remark}

\subsection{Analysis of the BFV \texorpdfstring{$L_\infty$}{L-infty} structure}\label{s:Linftyanalysis}
The BFV data for general relativity are encoded by a $(0)$-symplectic graded manifold (Equation \ref{e:BFVspaceEH}), endowed with a cohomological vector field. Graded manifolds endowed with a cohomological vector field correspond to ``higher'' Lie algebroid structures if the grading is concentrated in nonnegative degrees (see \cite{Vaintrob:1997} for the correspondence with Lie algebroids and \cite{VoronovQmanHigher} for the ``higher'' generalization).
However, the manifold $\calF_{\sfBFV}$ is concentrated in both positive and negative degrees (more precisely, in degrees $\{-1,0,1\}$). In this case, the correct notion is that of a strong homotopy Lie algebroid  (or $L_\infty$-algebroid). 

\begin{definition}[\cite{Vaintrob:1997,BruceInftyAlgds,BandieraChenStienonXuLiePairs}]
A (strict) $L_\infty$-algebroid is a vector bundle of $\mathbb{Z}$-graded manifolds $\mathsf{A}\to \mathcal{B}$ endowed with a cohomological vector field $Q\in{\calX}(\mathsf{A}[1])$ that is tangent to the zero section $0\colon \mathcal{B} \hookrightarrow \mathsf{A}[1]$.
\end{definition}

\begin{remark}
Notice that while in \cite{BruceInftyAlgds} a non-strict notion of a Lie algebroid is given (as a vector bundle $E\to M$ with a $Q$-structure on $E[1]$),
the strict notion requires tangency of $Q$ to the $0$-section.  In \cite{BandieraChenStienonXuLiePairs} the strict requirement is used implicitly, and the object is simply called an $L_\infty$-algebroid.
\end{remark}

It is easy to check that  the $Q$-manifold $(\mathcal{F}_{\sfBFV}^{EH}, Q_{\sfBFV}^{EH})$ indeed defines a (strict)  $L_\infty$-algebroid\footnote{This is to be expected from the general theory of BFV resolutions, see \cite{SchaetzBFV}.} given by the vector bundle
\begin{equation*}
    \mathsf{A}_{\sfBFV} \to \mathcal{B}_{\sfBFV},
\end{equation*}
where the total space and base are defined by:
\begin{align*}
\mathsf{A}_{\sfBFV}&\doteq\mathcal{V}(\Sigma) \times T^*\mathbfcal{R}(\Sigma) \times \Phi[-1](\Sigma) \\
\mathcal{B}_{\sfBFV}&\doteq T^*\mathbfcal{R}(\Sigma) \times \Phi[-1](\Sigma)
\end{align*}
so that the base is parametrized by $(h,\Pi,\chi_n,\chi_\partial)$ and the fibres are parametrized by a  function and  vector field  $(\phi,X)$ (both even). We parametrize the {\em shifted} fibres by $\xi^n,\xi^\partial$, and we can directly see that $\mathcal{F}_{\sfBFV}^{EH}=\mathsf{A}_{\sfBFV}[1]$.

The restriction to the zero section $0\colon \mathcal{B}_{\sfBFV}\hookrightarrow \mathsf{A}_{\sfBFV}$ is the cohomological vector field defined by:
\begin{align*}
    \underline{Q}(h) = \underline{Q}(\Pi) =0 ; \qquad \underline{Q}(\chi_\partial) = H_\partial(h,\Pi); \qquad  \underline{Q}(\chi_n) = H_n(h,\Pi),
\end{align*}
which represents the Koszul differential encoding the resolution of the vanishing ideal $\mathcal{I}_{EH}$ (Def. \ref{def:GRconstraints}). Equations \eqref{e:Qboundary-bracket1} and \eqref{e:Qboundary-bracket2}  define the bracket of \emph{constant} sections:
\begin{equation*}
    [(f_1,X_1),(f_2,X_2)]_{\mathsf{A}_{\sfBFV}} = (\Lie_{X_2}f_1 - \Lie_{X_1}f_2, [X_1,X_2]_{{\calX}(\Sigma)} + f_1 \mathrm{grad}_h\, f_2 - f_2\mathrm{grad}_h\, f_1),
\end{equation*}
where $X_1,X_2\in {\calX}(\Sigma)$ and $f_1,f_2 \in C^\infty(\Sigma)$. The anchor and multi-anchor map are encoded in the remainder of equations \eqref{e:Qboundary}. Indeed, Equation \eqref{e:Qboundary-Pi} reveals a quadratic dependency in the fibres due to the term $(\chi_\partial\otimes_s d\xi^n)^{\sharp\sharp}\xi^n$, which is then interpreted as a multi-anchor map, which on constant sections yields the vector field:
\begin{subequations}
\begin{align}
    &\rho^{(2)}((f_1,X_1),(f_2,X_2))(\Pi) = \chi_\partial^\sharp \otimes_s(f_1 \mathrm{grad}_h f_2 - f_2 \mathrm{grad}_h f_1 ), \label{e:Qtwoanchor}\\
    &\rho^{(2)}((f_1,X_1),(f_2,X_2))(h) =0 \\
    &\rho^{(2)}((f_1,X_1),(f_2,X_2))(\chi_\partial) =0, \\
    &\rho^{(2)}((f_1,X_1),(f_2,X_2))(\chi_n) =0.
\end{align}
\end{subequations}

Observe that equations \eqref{e:Qboundary-chiP} and \eqref{e:Qboundary-chin} instead contain constant terms in the fibre variables (respectively $H_\partial$ and $H_n$), which are then interpreted as anchors of arity\footnote{We recall that the arity of an operation is the number of arguments upon which it depends, see \cite{Shepard69Languages} for the origin of the term.} zero. In particular, we stress that the $L_{\infty}$ structure defined by $Q$ is not that of a Lie algebroid, owing to the presence of anchors of arity $0$ and $2$.

\begin{remark}\label{rem:restrictionargument}
Notice that the brackets on constant sections of $\mathsf{A}_{\sfBFV}$, which we extract from equations \eqref{e:Qboundary-bracket1} and \eqref{e:Qboundary-bracket2}, match the brackets of constant sections of the Lie algebroid of evolutions that two of the authors introduced with Fernandes in \cite[section 2.6]{BFW}. In that case, the base of the Lie algebroid is given by $\Sigma$-universes, i.e.\ isometry classes of (connected, lorentzian) manifolds $M$ together with an embedding of $\Sigma$ as a spackelike hypersurface.  Moreover, observe that Equation \eqref{e:Qboundary-h} can be taken as defining an ``action'' of a section of $\mathsf{A}_{\sfBFV}$ on $\mathbfcal{R}(\Sigma)$, which coincides with the anchor of the algebroid of evolutions, i.e.\ the action of gaussian vector fields on gaussian metrics, after restricting to $\Sigma$ (compare \eqref{e:Qboundary-h} with \cite[Equations 14 and 16]{BFW}). 
\end{remark}

At first glance, the $L_\infty$ structure defined by $Q^{EH}_{\sfBFV}$ appears to be lacking one piece of data: there is an anchor of arity two, but no higher brackets, a consequence of the absence of cubic terms in the $\xi$ variables in $Q^{EH}_{\sfBFV}$. In fact, though, from the defining Leibniz rule for three-brackets in $L_\infty$-algebroids, denoting constant sections by $s_i$, we get that 
\[
[s_1,s_2,f s_3]_{(3)} = \rho^{(2)}(s_1,s_2)(f) s_3 + f [s_1,s_2,s_3] = \rho^{(2)}(s_1,s_2)(f) s_3.
\]
This implies that while three-brackets of constant sections are indeed vanishing, there is a nontrivial three-bracket on nonconstant sections. 

The situation is akin to that in the more familiar example of the action algebroid $M\times \mathfrak{g}\to M$ for  an abelian Lie algebra $\mathfrak{g}$. Here, the bracket of constant sections is zero.  On the other hand, unless the action is trivial, the bracket is not identically zero on nonconstant sections.

\section{The Lie-Rinehart algebra}\label{s:newLR}
In the previous section, we constructed, over a graded manifold whose body is the cotangent bundle $T^*\mathbfcal{R}$ of the space of riemannian metrics on $\Sigma$, an $L_\infty$-algebroid whose bracket on constant sections coincides with the Poisson bracket on constraints for the initial value problem of general relativity.   

We would like to manipulate this data to try to obtain a structure that is more closely related to that of a Lie algebroid, motivated by existence of the Lie algebroid of evolutions, introduced by two of the authors in \cite{BFW}, on the (larger) space of $\Sigma$-universes (presentable in terms of paths of metrics $\mathcal{R}(\Sigma)^I$). In order to do this, in this section we will eliminate the higher order terms in the $L_\infty$ structure by replacing the variables with nonzero grading in the base with degree zero variables which are infinitesimal.  The price we will pay is that the object on whose sections the bracket is defined is no longer a vector bundle.  Thus, the resulting object is a Lie-Rinehart algebra slightly more general than that arising from a Lie algebroid.

Recall that the space of BFV fields is given by the graded manifold (Eq. \eqref{e:BFVspaceEH})
\begin{align}\label{boundaryfields2}
\mathcal{F}_{\sfBFV} 
    &= T^*\mathbfcal{R}(\Sigma) \times \mathcal{V}[1](\Sigma) \times 
        \Phi[-1](\Sigma) ,
\end{align}
and we have used variables $(h,{\Pi}, {\xi^\partial}, {\xi}^n,{\chi_\partial},{\chi}_n)$ to generate the algebra of smooth functions over $\mathcal{F}_{\sfBFV}$. Consider the new (degree $0$) variables\footnote{The $\chi$ variables are often referred to as ``antifields'' in physics terminology, while the $\xi$ variables are called ``ghosts''. Notice that here we are using only the ``ghosts for transversal diffeomorphisms'' $\xi^n$ to redefine the antifields.}
\begin{equation}
    {\psi}_n\doteq \chi_n\xi^n \qquad {\psi}_\partial\doteq\chi_\partial\xi^n.
\end{equation}
They satisfy $\psi_n^2=0$ and $\psi_\partial^2=0$, even though these are not odd variables, as well as ${\psi_n}{\psi_\partial}=0.$ In addition, we should  have $\psi_n\xi^n = \psi_\partial\xi^n=0$, since $(\xi^n)^2=0$.

More precisely, then, we consider as the algebra playing the role of functions on the base 
\[
   \mathfrak{B} \doteq C^\infty\left(T^*\mathbfcal{R}(\Sigma)\times  \Phi(\Sigma) \right) \slash{\langle \psi_\bullet^2\rangle},
\]
by which we mean the (commutative) algebra over $C^\infty\left(T^*\mathbfcal{R}(\Sigma)\right)$, freely generated by $\psi_\bullet\in\{\psi_\partial,\psi_n\}$, modulo the relations $\psi_n^2=0,$ $\psi_\partial^2=0$, and ${\psi_n}{\psi_\partial}=0$. (Note that the $\psi$'s are ``infinitesimal'' variables, but they are not odd.)  Thus, we may think of $\mathfrak{B}$ as the algebra of smooth functions on a ringed space $\mathbfcal{B}$ which is the product of the phase space $T^*\mathbfcal{R} (\Sigma)$ with the ``one point'' space which is the first infinitesimal neighborhood of the origin in $\Phi(\Sigma)=(\Omega^1 \otimes \mathcal{D})(\Sigma) \times \mathcal{D}(\Sigma)$. What will play the role of our Lie algebroid is a sheaf rather than a vector bundle over $\mathbfcal{B}$ (see Remark \ref{rmk:CategoricalInterpretation} for details). This means that we will have what is called a Lie-Rinehart algebra\footnote{The term Lie-Rinehart algebra was introduced by H\"ubschmann in \cite{Huebschmann90}.  See also \cite{huebschmann2004lie,Huebschmann21} and the historical remarks therein.} \cite{Rinehart} over $\mathfrak{B}$. We recall the definition:
\begin{definition}\label{def:LRAlg}
A Lie-Rinehart algebra is a pair $(\mathfrak{A},\mathfrak{m})$ consisting of a commutative algebra  $\mathfrak{A}$, a Lie algebra $(\mathfrak{m},[\cdot,\cdot])$ over $\mathbb R$, and a Lie algebra homomorphism $\rho_{\mathfrak{m}}\colon \mathfrak{m} \to \mathrm{Der}(\mathfrak{A})$, such that:
\begin{enumerate}
    \item $\mathfrak{m}$ is an $\mathfrak{A}$-module, and $\rho_{\mathfrak g}$ is a module homomorphism;
    \item for every pair of elements $X,Y\in\mathfrak{m}$ and $a\in\mathfrak{A}$, we have 
    \[
    [X,aY] = a[X,Y] + \rho(X)(a)Y \,.
    \]
\end{enumerate}
We also call $\mathfrak m$ a Lie-Rinehart algebra over $\mathfrak A$.
\end{definition}

Examples of Lie-Rinehart algebras over the algebra $\mathfrak{A}=C^\infty(M)$ of smooth functions on a manifold $M$ come from Lie algebroids.  Here, the Lie algebra $\mathfrak{m}$ is the space of smooth sections of a vector bundle ${E}\to M$ carrying a Lie algebroid structure, with its anchor map $\rho\colon E\to TM$ inducing the required module homomorphism. In Remark \ref{rmk:CategoricalInterpretation}, we will explain a slight generalization of this picture which covers our construction in general relativity.

In the case at hand, then, the algebra is $\mathfrak B$, and the  module is the quotient $\mathfrak L$ of the $\mathfrak{B}$ module $\mathcal{V}(\Sigma) \otimes \mathfrak{B}$ by the relations $\psi_n\xi^n = 0$ and $\psi_\partial\xi^n=0,$ where we have added generators $\xi^n,\xi^\partial$ for  the  components in the summands $C^\infty(\Sigma)$ and ${\calX}(\Sigma)$ of $\mathcal V (\Sigma)$ respectively; i.e.
\[
    \mathfrak{L}\doteq (\mathcal{V}(\Sigma) \otimes \mathfrak{B})/\langle\psi_\bullet \xi^n\rangle.
\]
This makes our module the sections of a sheaf over the ringed space $\mathbfcal B$, where the  component in $C^\infty(\Sigma)$ is supported along the vanishing locus of $\psi_n$ and $\psi_\partial$.

The Lie-Rinehart structure itself  is encoded by a differential $\tQ$ which is an odd derivation, squaring to zero, of the graded algebra:
\[
    \mathfrak{A}^{\text{gr}}\doteq C^\infty\left(\mathcal{V}[1](\Sigma) \right) \otimes \mathfrak{B}/\langle \psi_\bullet \xi^n\rangle.
\]
Writing the BFV differential {$Q^{EH}_{\sfBFV}$ }in our new variables,  we obtain (c.f.\ Equation \eqref{e:Qboundary})
\begin{subequations}\label{e:newQ}\begin{align}
\tQ(\xi^n)&=\Lie_{\xi^\partial}\xi^n \label{e:newQ-bracket1}\\
\tQ(\xi^\partial)&=\xi^n\, \mathrm{grad}_h\xi^n + \frac12[\xi^\partial,\xi^\partial]\label{e:newQ-bracket2}\\
\tQ(h)&=\widetilde{h}\xi^n -\Lie_{\xi^\partial}h = - 2K\xi^n - \Lie_{\xi^\partial} h \label{e:newQ-h}\\
\tQ(\Pi) &= \widetilde{\Pi}\xi^n +\Vol_h\left({G}^{\sharp\sharp}(h) \xi^n + {D}^{\sharp\sharp}_{h}(\xi^n)\right) - \Lie_{\xi^\partial}\Pi +  (\psi_\partial\otimes_s d\xi^n)^{\sharp\sharp}
\label{e:newQ-Pi}\\
\tQ(\psi_\partial)&= {H}_\partial\xi^n +  \Lie_{\xi^\partial}{\psi}_\partial - {\psi}_n d\xi^n \label{e:newQ-chiP}\\
\tQ(\psi_n)&={H}_n \xi^n  +  \Lie_{\xi^\partial}\psi_n  - 2 \Lie_{\psi_\partial^\sharp}\xi^n. \label{e:newQ-chin}
\end{align}\end{subequations}
We prove below that $\tQ$ is a derivation $\tQ$ of $\mathfrak{A}^{\text{gr}}$ that squares to zero.

Observe that Equation \eqref{e:newQ-chin} is a slight simplification of \eqref{e:Qboundary-chin}, as we compute
\begin{align*}
\tQ(\psi_n) &= Q(\chi_n) \xi^n + \chi_n Q(\xi^n) \\
& = H_n\xi^n + \Lie_{\xi^\partial}\psi_n - 2 \Lie_{\chi_\partial^\sharp}(\xi^n\Vol_h^{-\frac12}) \Vol_h^{\frac12}\xi^n \\
&= H_n\xi^n + \Lie_{\xi^\partial}\psi_n - 2(\Lie_{\eta^\sharp}\xi^n \Vol_h + \frac12 \xi^n\Div_h(\eta^\sharp)\Vol_h) \xi^n\\ 
& = H_n\xi^n + \Lie_{\xi^\partial}\psi_n - 2\Lie_{\eta^\sharp}\xi^n \Vol_h\xi^n\\
& =: H_n\xi^n + \Lie_{\xi^\partial}\psi_n - 2\Lie_{\psi_\partial^\sharp}\xi^n,
\end{align*}
where we used the same notation as in Remark \ref{rem:reparametrizationanddensities}.

We summarize the previous discussion with the following
\begin{theorem}\label{thm:LRalgebra}
$(\mathfrak{A}^{\text{gr}},\tQ)$ is a differential, graded algebra, and the data $(\mathfrak{B}, \mathfrak{L})$ define a Lie-Rinehart algebra. 
\end{theorem}

\begin{proof}
Consider the following diagram
\[
    \xymatrix{
    \mathfrak{A}^{\text{gr}} \ar[r]^-{\pi^*} \ar@{->}@/_1.7pc/[rr]_{\underline{\mathcal{f}^*}} & C^\infty\left(\mathcal{V}[1](\Sigma) \right) \otimes \mathfrak{B} \ar[r]^-{\mathcal{f}^*} & C^\infty(\calF_{\sfBFV}^{EH}) ,\\
    }
\]
where $\mathcal{f}^*$ is the (injective) map on functions induced by the (surjective) change of variables 
\[
(h,\Pi,\xi^\bullet,\chi_\bullet) \mapsto (h,\Pi,\xi^\bullet,\psi_\bullet = \chi_\bullet \xi^n),
\] 
and $\pi^*$ is induced by the quotient by the ideal $\langle\psi_\bullet\xi^n\rangle$.
By definition of $\tQ$, we have that $Q \circ \mathcal{f}^*=\mathcal{f}^* \circ \tQ$, where, here and below, we abbreviate $Q^{EH}_{\sfBFV}$  by $Q$ as in \eqref{e:Qboundary}. It follows that $\tQ$ is a derivation of $C^\infty(\mathcal{V}[1])\otimes\mathfrak{B}$, which we need to prove preserves the ideal $\langle\psi_\bullet \xi^n\rangle$:
\begin{align*}
\tQ(\psi_\partial \xi^n) &= \tQ(\psi_\partial) \xi^n + \psi_\partial \tQ(\xi^n) \\
    &= \left({H}_\partial\xi^n +  \Lie_{\xi^\partial}{\psi}_\partial - {\psi}_n d\xi^n\right)\xi^n + \psi_\partial \Lie_{\xi^\partial} \xi^n\\
    & = \Lie_{\xi^\partial}(\psi_\partial\xi^n) - \psi_n\xi^n d\xi^n = 0 \qquad (\text{mod}\ \psi_\bullet\xi^n) \\
\tQ(\psi_n \xi^n) &= \tQ(\psi_n) \xi^n + \psi_n \tQ(\xi^n) \\
    &= \left({H}_n \xi^n  +  \Lie_{\xi^\partial}\psi_n  - 2 \Lie_{\psi_\partial^\sharp}\xi^n\right)\xi^n + \psi_n\Lie_{\xi^\partial}\xi^n \\
    &= \Lie_{\xi^\partial}(\psi_n \xi^n) - 2\Lie_{\psi_\partial^\sharp \xi^n}\xi^n = 0 \qquad (\text{mod}\ \psi_\bullet \xi^n),
\end{align*}
where we used that $\xi^n$ squares to zero, we combined $(\Lie_{\xi^\partial}{\psi}_\bullet)\xi^n +  {\psi}_\bullet(\Lie_{\xi^\partial}\xi^n) = \Lie_{\xi^\partial}(\psi_\bullet\xi^n)$ (recall that $\psi_\bullet$ has degree $0$), and we used that $\psi_\partial^\sharp \xi^n = 0\, (\text{mod}\ \psi_\partial \xi^n)$. The generators $\psi_\bullet$, $\xi^n$ are to be thought as vector-valued---coordinate---functions, so that the expression $\Lie_{\xi^\partial}(\psi_\bullet\xi^n)$ denotes the Lie derivative on the \emph{values} of the function $\psi_\bullet\xi^n$ on $\mathbfcal{B}$. Modulo the ideal, that value---and thus its Lie derivative---vanishes.

$\tQ$ then descends to a derivation on $\mathfrak{A}^{\text{gr}}$, which we denote by another symbol $\underline{\tQ}$ for clarity, and we have $Q \circ \underline{\mathcal{f}}^*=\underline{\mathcal{f}}^* \circ \underline{\tQ}$.

We now prove that $\tQ$ squares to zero, and is thus a differential.  First we see that, on the generators $\xi^n,\xi^\partial,h$, we have $\tQ^2=0$ because $Q\equiv \tQ$ on these generators, and $Q^2=0$.
Moreover, on the new variables $\psi_\bullet$, we have
\begin{multline}\label{e:Qtildesquarestozero}
    \tQ^2(\psi_\bullet) = Q(Q(\chi_\bullet \xi^n)) = Q( Q(\chi_\bullet) \xi^n - \chi_\bullet Q(\xi^n)) \\
    = Q^2(\chi_\bullet) \xi^n + (-1)^{|Q|+|\chi_\bullet|}Q(\chi_\bullet) Q(\xi^n) - Q(\chi_\bullet) Q(\xi^n) + \chi_\bullet Q^2(\xi^n) = 0,
\end{multline}
where we used that $Q^2=0$ on $\calF_{\sfBFV}^{EH}$, and $|\chi_\bullet|=-|Q| = -1$. 
Then, we are left to check that $\tQ^2(\Pi)=0$. This follows from the fact that $Q ({\mathcal{f}}^*(\Pi))={\mathcal{f}}^* (\tQ(\Pi))$, and $Q^2 ({\mathcal{f}}^*(\Pi))=Q({\mathcal{f}}^* (\tQ(\Pi))) = {\mathcal{f}}^* (\tQ^2(\Pi))$, in virtue of the injectivity of $\mathcal{f}^*$.

This, together with the fact that $Q$ descends to the quotient by the ideal $\langle\psi_\bullet\xi^n\rangle$, allows us to show that $\underline{\tQ}$ also squares to zero on all generators, and $\underline{\mathcal{f}}^*$ is a chain map. Hence, from now on, we shall no longer distinguish between $\tQ$ and $\underline{\tQ}$.

\bigskip

Using the outputs of the differential $\tQ,$ we define the anchor $\rho\colon \mathfrak{L} \to \mathrm{Der}(\mathfrak{B})$ on generators by
\begin{subequations}
\begin{align*}
    \rho(\phi,X)(h) & = \widetilde{h}\phi - \Lie_X h \\
    \rho(\phi,X)(\Pi) & = \widetilde{\Pi}\phi + \Vol_h\left({G}^{\sharp\sharp}+{D}_h^{\sharp\sharp}(\phi)\right) - \Lie_X\Pi + \left(\psi_\partial\otimes_s d\phi\right)^{\sharp\sharp}\\
    \rho(\phi,X)(\psi_\partial) & = H_\partial\phi + \Lie_X\psi_\partial - \psi_n d\phi \\
    \rho(\phi,X)(\psi_n) & = H_n\phi + \Lie_X \psi_n - 2\Lie_{\psi_\partial^\sharp}\phi,
\end{align*}
\end{subequations}
where we used the shorthand notation introduced at the beginning of Section \ref{s:BFVGR} as well as Equations~\eqref{definitions}. The bracket structure on $\mathfrak{L}$ is given by 
\begin{equation*}
    [(\phi_1,X_1),(\phi_2,X_2)]_{\mathfrak{L}}=\left(\Lie_{X_1}\phi_2 - \Lie_{X_2}\phi_1, [X_1,X_2]_{{\calX}(\Sigma)} + \phi_1\Grad_h\phi_2 - \phi_2 \Grad_h \phi_1\right),
\end{equation*}
which is precisely that of the constraints for the initial value problem.
That these operations satisfy the properties of a Lie-Rinehart algebra, as defined in Definition \ref{def:LRAlg}, follows from the fact that $\tQ^2=0$, using a direct extension from Lie algebroids to Lie-Rinehart algebras of the argument in \cite{Vaintrob:1997}.
\end{proof}

\begin{remark}
\label{rmk:CategoricalInterpretation}
We may think of $\mathbfcal{B}$, defined above, as the product of the smooth manifold $T^*\mathbfcal R$ and an ``infinitesimally thickened point''.  The Lie-Rinehart algebra $\mathfrak L$ may then be thought of as sections of a geometric object $\mathbfcal L$ over $\mathbfcal B$, but this object is not a vector bundle, since, as a module over $\mathfrak B$, it is not locally free (with respect to the topology coming from the body, $T^* \mathbfcal R$), thanks to the relations $\psi_\bullet \xi^n = 0.$ Rather, this object is a more general ``sheaf'', since its sections, the elements of $\mathfrak L$, can be localized, again with respect to the topology coming from $T^* \mathbfcal R$. 

In order to make this notion more precise, one could work in the category of smooth loci, or $C^\infty$ schemes, or synthetic differential geometry \cite{DubucCahiers,KockConvenient,KockReyesCorr,joyce2019algebraic}. We will not detail such a description, as it lies outside the scope of the present paper.
\end{remark}

\begin{remark}\label{rmk:nonHamiltonianLRalg}
Our generalized Lie algebroid structure is not hamiltonian, nor even presymplectically anchored in the sense of \cite{BlohmannWeinstein:HamLA}, with respect to the pre-symplectic structure given by the canonical form on $T^*\mathbfcal{R}(\Sigma)$ pulled back to the base $\mathbfcal{B}$ of our ``algebroid''. One way to see this is that the leaves of the foliation given by the projection to $T^*\mathbfcal{R}(\Sigma)$ are the characteristics of the pre-symplectic form. If the structure were presymplectically anchored, the anchor would descend to the symplectic quotient (at least when evaluated on constant sections). However, looking at the term dependent on $\psi_\partial$ in \eqref{e:newQ-Pi}, we see that the fibres with $h$ and $\Pi$ constant are not preserved.
 
\end{remark}

\section{Alternative approaches}\label{s:Alternative}

This section is independent of Section \ref{s:newLR}. In Section \ref{s:projectiontoinitialdata} we will pursue an alternative investigation, aimed at constructing an algebroid on $T^*\mathbfcal{R}$ whose bracket of constant sections reproduces Equation \eqref{e:ConstraintBrackets}, starting again from the BFV data of Theorem \ref{t:BVBFVEH}. In Section \ref{s:unsuccessful} we will analyze the limitations of two attempts to construct a Lie algebroid on any symplectic manifold given a coisotropic submanifold.

\subsection{Reducing the BFV structure to \texorpdfstring{$T^*\mathbfcal{R}(\Sigma)$}{T*R}}\label{s:projectiontoinitialdata}

Following the observations in Remark \ref{rem:restrictionargument}, since equations \eqref{e:Qboundary-bracket1} through \eqref{e:Qboundary-h} are closely related to the algebroid of evolutions of \cite{BFW}, it would seem reasonable to attempt to get rid of the problematic terms involving the $\chi$ variables.

\begin{remark}\label{rem:gaugeGRcomparison}
This point of view makes sense if we recall that the $L_\infty$ structure given by the BFV theory arises as the combination of Koszul and Chevalley--Eilenberg resolutions. In the standard scenario of gauge theory, this combination is simple, i.e.\ we have a decomposition $Q_{\sfBFV} = d_{\mathfrak{g}} + d_K$, where $\check{d}_{\mathfrak{g}}$ is the Chevalley--Eilenberg differential associated with a Lie algebra action on the body of the BFV space of fields. An observation that hints at a possible structural similarity between gauge theories and GR is that the BFV action is linear in the degree $-1$ variables $\chi_n,\chi_\partial$ (see also \cite{FradkinVilkoviskyGrav}). We can thus hope that an analogous decomposition holds here as well. Unfortunately, as we will see, this is not the case.
\end{remark}

The fact that equations \eqref{e:Qboundary-chin} and \eqref{e:Qboundary-chiP} contain terms that are inhomogeneous in $\chi_n$ and $\chi_\partial$ means that $Q_{\sfBFV}$ is not tangent to the zero section of the vector bundle\footnote{Observe that this is to be expected even in standard gauge-theory scenarios.} 
\begin{equation}
\label{eq:F0Bundle}
\calF_{\sfBFV} \to \mathcal{F}_0 \doteq T^*\mathbfcal{R}(\Sigma) \times \mathcal{V}[1](\Sigma),    
\end{equation}
defined by $\chi_n=\chi_\partial=0$. Hence, the BFV data of Section \ref{s:gr}, Theorem \ref{t:BVBFVEH} do not restrict to $\mathcal{F}_0$, and we will need to \emph{project} onto it. 

As is the case for any vector bundle, the tangent space at the zero section of~\eqref{eq:F0Bundle} splits naturally into a vertical and a horizontal part. The vertical part of $Q_{\sfBFV}$ is given by $Q(\chi_n)$ and $Q(\chi_\partial)$ in~\eqref{e:Qboundary}. The horizontal part at the zero section is the degree 1 vector field on $\mathcal{F}_0$ given by

\begin{subequations}\label{Q0}\begin{align}
Q_0(\xi^n)&=L_{\xi^\partial}\xi^n \label{Q0b1}\\
Q_0(\xi^\partial)&=\xi^n\, \mathrm{grad}_h\xi^n + \frac12[\xi^\partial,\xi^\partial]\label{Q0b2}\\
Q_0(h)&=\widetilde{h}\xi^n + L_{\xi^\partial}h = -2K\xi^n + L_{\xi^\partial} h\label{Q0a1}\\
Q_0(\Pi) &= \widetilde{\Pi}\xi^n + \sqrt{h}{G}^{\sharp\sharp}(h) \xi^n + \sqrt{\mathrm{h}}{D}^{\sharp\sharp}_{h}(\xi^n) + L_{\xi^\partial}\Pi. \label{Q0a2}
\end{align}\end{subequations}
We want to think of $\mathcal{F}_0$ itself as the shifted, trivial, vector bundle:
\begin{equation*}
    \mathcal{F}_0\equiv\mathsf{A}_0[1] \doteq T^*\mathbfcal{R}(\Sigma) \times \mathcal{V}[1](\Sigma) \to T^*\mathbfcal{R}(\Sigma),
\end{equation*}
and establish whether the vector field $Q_0$ endowes the vector bundle $\mathsf{A}_0$ with a Lie algebroid structure. As we did earlier, we can extract information on the bracket of (constant) sections from equations \eqref{Q0b1} and \eqref{Q0b2}, while we can define an anchor for the vector bundle $\mathsf{A}_0$ using equations \eqref{Q0a1} and \eqref{Q0a2}. However, this defines a Lie algebroid structure if and only if $Q_0^2=0$.

\begin{proposition}\label{prop:Q0prop}
With the definitions above we have $[Q_0,Q_0]\not=0$, in particular:
\begin{equation*}
Q_0^2(\xi^n)=Q_0^2(\xi^\partial)=Q_0^2(h)=0; \qquad    Q_0^2(\Pi)=  H_\partial \otimes d\xi^n\,\xi^n \not=0.
\end{equation*}
\end{proposition}
\begin{proof}
We abbreviate $Q_{\sfBFV}^{EH}$ by $Q$ for ease of notation. It is immediate to check that\footnote{We omit an obvious pullback along $\calF_{\sfBFV}\to\calF_0$.} $Q(\xi^n)=Q_0(\xi^n)$, and similarly for $\xi^\partial$ and $h$. On the other hand, we have that, expanding functions as polynomials in $\chi_\partial$,
$$
Q (\Pi) = Q(\Pi)_0 + Q(\Pi)_1 = Q_0 (\Pi) - \chi_\partial \otimes_s d\xi^n\, \xi^n,
$$
where $Q(\Pi)_0= Q_0(\Pi)$ and $Q(\Pi)_1 = - \chi_\partial \otimes_s d\xi^n\, \xi^n$ are, respectively, inhomogeneous and linear in $\chi_\partial$.  For any function $f$ on $\mathcal{F}_0$, we have that 
$$
Q(f) = Q(f)_0 + Q(f)_1= Q_0(f) + Q(f)_1
$$

It is easy to gather that $Q_0^2(\xi^\partial)=0=Q_0^2(\xi^n)$ follows from directly from the vanishing of $Q^2$. To show that $Q_0^2(h) = 0$ we first observe that, since $Q_0(h)$ is a function on $
\calF_0$ we can apply the formula above to yield:
\begin{equation*}
    0=Q^2(h) = Q(Q_0(h)) = Q(Q_0(h))_1 + Q_0^2(h)
\end{equation*}
However, $Q_0^2(h)$ is the only $\chi_\partial$-inhomogeneous term in $Q^2(h)$, and therefore both terms vanish independently if $Q^2(h)=0$.
More explicitly,
$$
Q(Q_0(h))_1 = \frac{\delta Q_0(h)}{\delta \Pi} Q(\Pi)_1 = -2\frac{\delta (K \xi^n)}{\delta \Pi} \chi_\partial \otimes_s d\xi^n\, \xi^n \propto \xi^n{}^2 =0,
$$
so that $Q_0^2(h)=0$. 

Finally, let us decompose $Q(\chi_\partial)$, given by Equation~\eqref{e:Qboundary-chiP}, as a polynomial in $\chi_\partial$
$$
Q(\chi_\partial) = H_\partial + \Lie_{\xi^\partial}{\chi}_\partial - {\chi}_n d\xi^n = H_\partial + Q(\chi_\partial)_1.
$$
Since we know that $Q$ is cohomological, we have ($Q_0(\Pi)$ is a function on $\mathcal{F}_0$):
\begin{equation*}
\begin{split}
    0&=Q (Q(\Pi)) = Q(Q_0(\Pi))  -  Q(\chi_\partial \otimes_s d\xi^n\,\xi^n) \\
    &= Q_0(Q_0(\Pi)) + Q(Q_0(\Pi))_1  - H_\partial \otimes_s d\xi^n\,\xi^n  - Q(\chi_\partial)_1\otimes_s d\xi^n\,\xi^n  + \chi_\partial \otimes_s Q_0(d\xi^n\,\xi^n).
\end{split}
\end{equation*}
Thus, collecting terms in the polynomial expansion in $\chi$ we have
\begin{align*}
    Q^2(\Pi)=0 \iff 
    \begin{cases}
    Q_0^2(\Pi) = H_\partial \otimes d\xi^n\,\xi^n  \\
    Q(Q_0(\Pi))_1 - Q(\chi_\partial)_1\otimes_s d\xi^n\,\xi^n + \chi_\partial Q_0(d\xi^n\,\xi^n) =0
    \end{cases}
\end{align*}
so that $Q_0^2(\Pi)\not=0$, since $H_\partial \otimes d\xi^n\,\xi^n$ is not zero.

\end{proof}

\begin{remark}
The geometric interpretation of $Q_0$ is simple: $Q_0(h)$ and $Q_0(\Pi)$ represent the action of the hamiltonian vector fields of the constraint functions  $\mathbb{H}_n(\xi^n)$ and $\mathbb{H}_\partial(\xi^\partial)$ on the base variables $(h,\Pi)$. There is a close relation between $Q_0$ and the Lie algebroid of evolutions of \cite{BFW}. In the latter case, one obtains the data specified by $Q_0$, i.e.\ a bracket of sections for $\mathsf{A}_0$ and a vector bundle morphism $\mathsf{A}_0 \to T T^*\mathbfcal{R}(\Sigma)$ by  ``projecting'' the algebroid of evolutions to initial data. As mentioned without proof in \cite{BFW}, this projection 
does not give a Lie algebroid structure; this fact is testified to by the fact that $Q_0$ is not cohomological.

\end{remark}

In physical terms, the fact that $Q_0^2\neq 0$  means that the gauge transformations generated by the ``charges'' $\mathbb{H}_n(\xi^n)$ and $\mathbb{H}_\partial(\xi^\partial)$ only form an algebroid \emph{on-shell}, i.e.\ on the zero locus $C_{EH}$. Outside of $C_{EH}$ we have that the Jacobiator of three (non-constant) sections will not vanish, but it will be $Q$-closed. Notice, though, that the vector field $Q_{\sfBFV}$ restricts (i.e.\ is tangent) instead to the submanifold of $\mathcal{F}_0$ defined by the vanishing locus $\mathcal{I}_{EH}$, i.e.\ 
$$
\mathsf{A}_{\text{small}}[1]\doteq C_{EH} \times \mathcal{V}[1](\Sigma).
$$
Proposition \ref{prop:Q0prop} then tells us that the restriction of $Q_{\sfBFV}$ to $\mathsf{A}_{\text{small}}$ is cohomological. This defines the algebroid 
\begin{equation}
\xymatrix{
    \mathsf{A}_{\text{small}}\ar[rr]^-{\rho_{\text{small}}} \ar[dr] &&  \ar[dl]TC_{EH} \\
    & C_{EH} &
    }
\end{equation}
whose image under the anchor is given by the characteristic distribution of $C_{EH}$. Compare this construction with those analyzed in Section~\ref{s:unsuccessful}, below.

\begin{remark}
We defined the submanifold $\mathsf{A}_{\text{small}}=C_{EH}\times (0,0) \times\mathcal{V}[1](\Sigma)\subset \calF_{\sfBFV}$ by setting $\chi_\partial=\chi_n=0$ (i.e.\ restricting to $\calF_0$) and imposing all constraints. If we defined $\mathsf{A}'_{\text{small}}$ without requiring $H_n=0$, as naively suggested from the fact that only $H_\partial=0$ is needed to make $Q_0$ cohomological, Equation \eqref{e:Qboundary-chin} would tells us that $Q$ is not tangent to $\mathsf{A}'_{\text{small}}$, due to the term $H_n$, since we should have $Q(\chi_n)\vert_{\mathsf{A}'_{\text{small}}}=0$.
\end{remark}

\subsection{Other approaches for coisotropic submanifolds}\label{s:unsuccessful}

We explain in this section how two attempts to construct Lie algebroids over a symplectic manifold with a coisotropic submanifold, without adding extra variables, do not actually work and cannot be applied to the case at hand.\footnote{We found examples of these approaches in remarks in \cite{BojowaldDeformations} and an early version of \cite{cattaneo2021constrained}.}

Consider a symplectic manifold $M=T^*N$ (for simplicity), and let $\{H_i\}_{i=1}^n$ denote a set of $n$ functions in involution: $\{H_i,H_j\}=f_{ij}^k(x)H_k$ with $f_{ij}^k(x)\in C^\infty(M)$ the structure functions (we understand summation over repeated indices). Denote by $X_i$ the hamiltonian vector field associated to $H_i$, i.e.\ $X_i=\{\cdot,H_i\}$. Observe, \emph{en-passant}, that the hamiltonian vector fields are in involution only ``on-shell'', i.e.\ on the vanising locus of the $H_i$'s, unless the structure functions are constant; i.e.\ for any function $g\in C^\infty(M)$: 
\begin{equation}\label{e:notininvolution}
    [X_i,X_j](g) = f_{ij}^k X_k(g) + \{f_{ij}^k,g\} H_k 
\end{equation}

Consider now a rank-$n$ trivial vector bundle $A\to M$, and decompose sections $s\in\Gamma(A)$ as $s= s^i u_i$ for $\{u_i\}_{i=1}^n$ a basis of constant sections. We try to endow this vector bundle with a Lie algebroid structure.

\subsubsection{Alternative 1 (see \cite[Sect. 4.3]{cattaneo2021constrained})}
We define the anchor map as a vector bundle morphism $\rho: A\to TM$ as follows:
\[
s \longmapsto \rho_1(s) \doteq s^i X_i = s^i \{\cdot,H_i\},
\]
and we declare a bracket on sections to be given by 
\[
    \llb s_1,s_2\rrb_1 \doteq \left(f_{ij}^k s^i_1s^j_2 + s^i_1X_i(s_2^k) - s^i_2X_i(s_1^k) \right)u_k.
\]
It is easy to check that with these definitions we have the Leibniz rule:
\[
    \llb s_1, g\, s_2\rrb_1 = g \llb s_1, s_2\rrb_1 + \rho_1(s_1)(g) s_2,
\]
however, with a simple calculation we see that 
\[
\rho_1(\llb s_1,s_2\rrb_1) - [\rho_1(s_1),\rho_1(s_2)] = \dots = f_{ij}^ks_1^is_2^j X_k - s_1^is_2^jX_iX_j + s_2^is_1^jX_iX_j
\]
applying the above vector field to any function $g$, thanks to Equation \eqref{e:notininvolution} we get
\begin{equation*}
    \left(\rho_1(\llb s_1,s_2\rrb_1) - [\rho_1(s_1),\rho_1(s_2)]\right)(g) = - s_1^i s_2^j\{f_{ij}^k,g\}H_k 
\end{equation*}
which vanishes if and only if either
\begin{enumerate}
    \item $f_{ij}^k$ constant, or
    \item we restrict to the locus $H_k=0$, i.e. on shell.
\end{enumerate}
\subsubsection*{Conclusion}
The prescription above yields an algebroid structure on $A$  only on shell, unless the structure functions are constant, in which case we have the usual action algebroid for a Lie algebra action.

\subsubsection{$Q$-manifold reformulation of Alternative 1}
The data above can be encoded by a degree 1 vector field $Q$ on $A[1]$, as follows:
\[
    Qc^i = f^i_{jk} c^jc^k = [c,c]^i ,\qquad   Qx^\mu = X^\mu_i c^i = X_i(x^\mu)c^i.
\]
This defines an algebroid iff $Q^2=0$ \cite{Vaintrob:1997,VoronovQmanHigher}. However, applying $Q^2$ to any function $g=g(x)$ gives
\[
Q^2g = X_j^\mu\partial_\mu X_i(g) c^j c^i - \frac12 X_i(g) f^i_{jk}c^jc^k = \{\{g,H_i\},H_i\}c^jc^i - \frac12\{g,H_k\}f^k_{ji}c^jc^i,
\]
which due to the antisymmetry of the product $c^i c^j$ of odd variables yields $Q^2(g) = \{g,f_{ji}^k\}H_k c^jc^i,$ which, as above, does not vanish off shell for all $g$ unless the structure functions  $f_{ji}^k$ are constant.
Moreover, we also have that
\[
Q^2c^i = [c,[c,c]]^i - \frac12 X^\mu_i\partial_\mu f^{i}_{jk}c^jc^k = \{f^i_{jk},H_i\}c^jc^k
\]
which again does not vanish in general.

\begin{remark}
Observe that this is the scenario we outlined when looking at $Q_0$ in the first part of Section \ref{s:projectiontoinitialdata}, which can be interpreted as a particular example of a general fact. Notice that, in the specific example of GR, the combination $\{f^i_{jk},H_i\}c^jc^k=0$ ought to vanish as a result of the first statement in Proposition \ref{prop:Q0prop}, i.e.\ $Q_0^2(\xi^n)=Q_0^2(\xi^\partial)=0$. 
\end{remark}

\subsubsection{Alternative 2, (see \cite[Section II.B]{BojowaldDeformations})}
With a similar starting point as above\footnote{In \cite{BojowaldDeformations} a nontrivial bundle $A$ is considered, in principle, but this will have no bearing on our observations.} we define an ``anchor'' map $\rho_2\colon \Gamma(A) \to {\calX}(M)$ as
\[
s\longmapsto \rho_2(s) \doteq \{s^iH_i, \cdot\}.
\]
One immediately sees that this map does not arise from a vector bundle morphism, as it is not $C^\infty$-linear (i.e.\ $\rho(gs)\not=g\rho(s)$), and hence the vector bundle $A\to M$ cannot be a Lie algebroid with anchor $\rho_2$. For the same reason, this map also fails to be a module homomorphism, so this is not the anchor map of a Lie--Rinehart algebra either (cf.\ Def.\ \ref{def:LRAlg}). However, as shown in the paper, with the bracket given as
\[
    \llb s_1,s_2\rrb_2 \doteq \{s_1^iH_i, s_2^j H_j\}
\]
one can check that both Leibniz and a (formal) homomorphism property are satisfied:
\[
    \llb s_1, gs_2\rrb_2 = g\llb s_1,s_2\rrb_2 + \rho_2(s_1)(g)s_2, \qquad \rho_2(\llb s_1,s_2\rrb_2) - [\rho_2(s_1),\rho_2(s_2)] = 0.
\]

\subsubsection*{Conclusion}
This structure is not that of a Lie algebroid, nor that of a Lie-Rinehart algebra. It might be given some other interpretation, once the appropriate notion is found, but we are not aware of any. Notice that our conclusions are independent of the constancy of the structure functions $f_{ij}^k$, so the construction in \cite{BojowaldDeformations} fails to yield a Lie algebroid even in the case of a hamiltonian action of a Lie algebra on $M$, where the functions $H_i$ are interpreted as the components of a momentum map.

\section{Outlook}
We conclude with a few observations on possible developments and further investigations aimed at a better understanding of the hamiltonian structure of general relativity.

\label{s:Outlook}
\subsection*{Cauchy data as composition of coisotropic relations}
Recall the construction of initial data for a lagrangian field theory outlined in Section \ref{s:LFTBoundary}. If we now assume that the equations admit solutions for a small cylinder  $M_{\epsilon}=\Sigma\times[0,\epsilon]$, the projection $L_{M_{\epsilon}}=\pi_{M_{\epsilon}}(EL)\subset \Fclp_{\partial M_{\epsilon}} = \overline{\Fclp_\Sigma}\times \Fclp_\Sigma$ is expected to be a lagrangian submanifold\footnote{Note that $L_{M_\epsilon}$ is always isotropic.}, 
which can be seen as the boundary values of a solution of the equations of motion in $M_{\epsilon}$ (think of the graph of the hamiltonian flow). A solution is found by intersecting $L_{M_{\epsilon}}$ with another lagrangian submanifold transversal to it: a boundary condition\footnote{Isotropicity of $L_M$ is generally insufficient to hope for the existence of a boundary condition transversal to it. For equations of motion that admit unique solutions in a short cylinder $M_\epsilon$, we expect $L_{M_\epsilon}$ to be lagrangian, in the sense of being isotropic  with an isotropic complement. More generally, $L_{M_\epsilon}$ might still be lagrangian in some weaker sense, and generally it might not be given by the graph of a hamiltonian flow. See \cite[Theorem 4.2]{CattaneoMnevWaveRel} for an in-depth analysis of the example of wave equations.}  (see \cite[Section 2.1]{CMRcorfu} and also \cite[Chapter III.16]{KiT1979}).

One can look at the subset $C\subset \Fclp_\Sigma$ of points that can be completed to a pair in $L_{M_\epsilon}\subset \overline{\Fclp_\Sigma}\times \Fclp_\Sigma$ for \emph{some} $\epsilon$. Let us assume that, for said $\epsilon$, the submanifold $L_{M_\epsilon}$ is lagrangian as expected.
If we think of $\Fclp_\Sigma$ as a relation from the point to $\Fclp_\Sigma$ itself---hence coisotropic---and $L_{M_\epsilon}$ as a lagrangian relation in $\overline{\Fclp_\Sigma}\times\Fclp_\Sigma$ (in particular coisotropic), following \cite[Section 3.4]{CMRcorfu}, this leads us to conclude that $C$ must be coisotropic. Indeed $C$ is the composition of coisotropic relations $C=\bigcup_{\epsilon\in[0,\infty)}L_{\Sigma\times[0,\epsilon]}\circ \Fclp_\Sigma$.

In order to have a proof of the coisotropic property of $C_{EH}$ that does not explicit rely on diffeomorphism symmetry, one could attempt to prove that, for a sufficiently small cylinder $\Sigma\times[0,\epsilon]$, the subset $L_{M_\epsilon}$ is actually a (locally) lagrangian submanifold.\footnote{This means that it is would  be lagrangian at smooth points.} One route to this goal, compatible to the discussions contained in this paper, follows the reduction procedure culminating in Corollary 3.21 of \cite{CattaneoMnevReshetikhin2014}.

Another interesting approach to an abstract proof of the coisotropic property,  unrelated to the present discussion, was recently proposed in \cite{Glowacki}. For the author, this property is a consequence of the requirement that the canonical and covariant approaches to the phase space for general relativity be physically equivalent.

\subsection*{The homotopy momentum map}

In recent work of one of the authors \cite{blohmann2021homotopy}, it was shown that the the non-integrated constraints of general relativity can be interpreted as a momentum map in the setting of multisymplectic geometry.

The multipresymplectic form of a lagrangian field theory such as general relativity is the $(n+1)$-form $\omega = \mathsf{el} + \delta\gamma$, the sum of the Euler-Lagrange form and the variation of the boundary form, which is the total derivative of the Lepage form $L + \gamma$. From $\omega$ we can construct an $L_\infty$-algebra $L_\infty(J^\infty \Lor_M,\omega)$ consisting of forms in the variational bicomplex, which can be interpreted as the $L_\infty$-algebra of conserved currents \cite[Prop.~2.4]{blohmann2021homotopy}.

The homotopy momentum map of a Lie algebra action $\rho:\mathfrak{g} \to \calX(J^\infty \Lor_M)$ is a morphism of $L_\infty$-algebras, 
\begin{equation*}
  \mu: \mathfrak{g} \to L_\infty(J^\infty \Lor_M,\omega)
  \,,
\end{equation*}
such that the first component maps every $a \in \mathfrak{g}$ on $M$ to a form $\mu_1(a) \in \Omega^{n-1}(J^\infty \Lor_M)$ satisfying
\begin{equation*}
  \iota_{\rho(a)} \omega = -\mathbf{d}\mu_1(a)
  \,,
\end{equation*}
where $\mathbf{d}$ is the total differential of the variational bicomplex.

The action of diffeomorphisms on Lorentz metrics is local, so that it induces an action of the Lie algebra of vector fields $\rho: \calX(M) \to \calX(J^\infty \Lor_M)$. It is shown in \cite[Thm.~4.1]{blohmann2021homotopy} that the Lepage form $L + \gamma$ is invariant under this action. Since the Lepage form is the primitive of the multipresymplectic form, it follows in analogy to symplectic geometry that the maps
\begin{equation*}
\begin{aligned}
  \mu_k: \wedge^k \calX(M)
  &\longrightarrow L_\infty(J^\infty \Lor, \omega)
  \\
  \mu_k(v_1 \wedge \ldots \wedge v_k) 
  &:= \iota_{\rho(v_1)} \cdots \iota_{\rho(v_k)} (L + \gamma)
  \,,
\end{aligned}
\end{equation*}
where $1 \leq k \leq n-1$, define a homotopy momentum map \cite[Theorem 4.2]{blohmann2021homotopy}.

The higher momentum map does not depend on the choice of the Cauchy hypersurface $\Sigma$. It is an open question, how the $L_\infty$-algebra of conserved currents is related by integration over $\Sigma$ to the Poisson algebra of the corresponding ``charges'', i.e.~the constraint functions.

\sloppy
\printbibliography

\end{document}